\pdfoutput=1
\documentclass[sigplan,screen,authorversion,balance=true]{acmart} 
\setcopyright{rightsretained}
\acmPrice{}
\acmDOI{10.1145/3497775.3503690}
\acmYear{2022}
\copyrightyear{2022}
\acmSubmissionID{poplws22cppmain-p38-p}
\acmISBN{978-1-4503-9182-5/22/01}
\acmConference[CPP '22]{Proceedings of the 11th ACM SIGPLAN International Conference on Certified Programs and Proofs}{January 17--18, 2022}{Philadelphia, PA, USA}
\acmBooktitle{Proceedings of the 11th ACM SIGPLAN International Conference on Certified Programs and Proofs (CPP '22), January 17--18, 2022, Philadelphia, PA, USA}



\usepackage[utf8]{inputenc}
\usepackage[T1]{fontenc} 

\usepackage{microtype}
\newcommand{\cal}{\mathcal}
\newcommand{\basedir}{tex/}
\makeatletter%
\@ifundefined{basedir}{%
\newcommand\basedir{}%
}{}%
\makeatother%

\usepackage{amsfonts}
\usepackage{amsthm}
\usepackage{stmaryrd}
\usepackage{bm}

\usepackage{mathpartir}
\usepackage{listings}
\usepackage{\basedir pftools}
\usepackage{\basedir iris}
\usepackage{\basedir lstcoq}
\usepackage{bussproofs}

\usepackage{xcolor}

\usepackage{graphicx}
\usepackage{enumitem}
\usepackage{semantic}
\usepackage{csquotes}
\usepackage{mdframed}

\usepackage{semantic}
\usepackage{hyperref}
\usepackage[nameinlink]{cleveref}

\definecolor{LinkColor}{rgb}{0.55,0.0,0.3}
\definecolor{CiteColor}{rgb}{0.55,0.0,0.3}

\hypersetup{%
  linktocpage=true, pdfstartview=FitV,
  breaklinks=true, pageanchor=true, pdfpagemode=UseOutlines,
  plainpages=false, bookmarksnumbered, bookmarksopen=true, bookmarksopenlevel=3,
  hypertexnames=true, pdfhighlight=/O,
  colorlinks=true,linkcolor=LinkColor,citecolor=CiteColor,
  urlcolor=LinkColor
}

\renewcommand{\Prop}{\textdom{Prop}}
\newcommand{\Set}{\textdom{Set}}
\newcommand{\Bunch}{\textdom{Bunch}}
\newcommand{\Atom}{\textdom{Atom}}
\newcommand{\Frml}{\textdom{Frml}}

\newcommand{\mult}{\cdot}
\newcommand{\empM}{\varnothing_m}
\newcommand{\empA}{\varnothing_a}
\newcommand{\ccomma}{\mathbin{\scaleobj{1.8}{\bm{,}}}}
\newcommand{\csemic}{\mathbin{\scaleobj{1.5}{\bm{;}}}}
\newcommand{\outI}[1]{\langle #1 \rangle}
\newcommand{\cl}[1]{\mathsf{cl}(#1)}
\newcommand{\frmlI}[1]{(#1)^{\bm{\ast}}}

\newcommand{\provesCF}{\proves_{\mathsf{cf}}}
\newcommand{\bunchDecomp}[3]{#1 \rightsquigarrow \langle #2 \mid #3 \rangle}

\newcommand{\longhash}{93aa9549cd122330d9cdb38900c69c21d2f5380a}
\newcommand{\shorthash}{93aa954}
\newcommand{\coqdocbasebaseurl}{https://co-dan.github.io/BI-cutelim/\shorthash}

\newcommand{\coqdocbaseurl}{\coqdocbasebaseurl /bunched.}
\newcommand{\urlhash}{\#}
\newcommand{\coqdocurl}[2]{\coqdocbaseurl #1.html\urlhash #2}
\newcommand{\nolinkcoqident}[1]{\texttt{\detokenize{#1}}}
\makeatletter
\newcommand{\coqident}{\begingroup\@makeother\#\@coqident}
\newcommand{\@coqident}[3][]{%
  \ifthenelse{\isempty{#2}}%
  {\nolinkcoqident{#3}}%
  {\ifthenelse{\isempty{#1}}%
  {\href{\coqdocurl{#2}{#3}}{\nolinkcoqident{#3}}}%
  {\href{\coqdocurl{#2}{#3}}{\nolinkcoqident{#1}}}}%
\endgroup}
\newcommand{\coqfile}[2]{%
  \ifthenelse{\isempty{#1}}%
  {\href{\coqdocbaseurl #2.html}{\nolinkcoqident{#2.v}}}%
  {\href{\coqdocbaseurl #1.#2.html}{\nolinkcoqident{#2.v}}}}
\newcommand{\coqmod}[3]{%
  \href{\coqdocurl{#1}{#2.#3}}{\nolinkcoqident{#3}}}
\makeatother

\begin{document}

\author{Dan Frumin}
\affiliation{\institution{Bernoulli Institute, University of Groningen}
   \country{The Netherlands}}
\email{d.frumin@rug.nl}

\date{\today}

\title[Semantic Cut Elimination for BI, Formalized in Coq]
{Semantic Cut Elimination for the Logic of Bunched Implications, Formalized in Coq}

\begin{abstract}
The logic of bunched implications (BI) is a substructural logic that forms the backbone of separation logic, the much studied logic for reasoning about heap-manipulating programs.
Although the proof theory and metatheory of BI are mathematically involved, the formalization of important metatheoretical results is still incipient.
In this paper we present a self-contained formalized, in the Coq proof assistant, proof of a central metatheoretical property of BI: cut elimination for its sequent calculus.

The presented proof is \emph{semantic}, in the sense that is obtained by interpreting sequents in a particular ``universal'' model.
This results in a more modular and elegant proof than a standard Gentzen-style cut elimination argument, which can be subtle and error-prone in manual proofs for BI.
In particular, our semantic approach avoids unnecessary inversions on proof derivations, or the uses of cut reductions and the multi-cut rule.

Besides modular, our approach is also robust:
we demonstrate how our method scales, with minor modifications, to (i) an extension of BI
with an arbitrary set of \emph{simple structural rules}, and (ii) an extension with an S4-like $\Box$ modality.
\end{abstract}

\maketitle

\section{Introduction}
  The logic of bunched implications (BI) \cite{ohearn:pym:99} is an extension of intuitionistic logic with substructural connectives.
  BI (and its classical cousin Boolean BI) is known for, among other things, forming a basis for separation {logic~\cite{reynolds:2002,ohearn:2019}} -- a popular program logic for verification of heap-manipulating programs.
  The BI itself, and many of its important models, are based on the idea that propositions denote ownership of resources and BI includes a \emph{separating conjunction} connective $\ast$, which signifies ownership of \emph{disjoint} resources~\cite{pym.etal:2004}.
  As an adjoint to $\ast$, BI also includes a \emph{magic wand} connective $\wand$, which is determined by the property
  \[
    A \proves B \wand  C \qquad \iff \qquad A \ast B \proves C.
  \]
  Additionally, BI includes a unit element $\EMP$ for the separating conjunction $\ast$.

Proof theoretically, BI can be formalized in a Gentzen-style sequent calculus, which operates on the judgments of the form $\Delta \proves A$,
where $\Delta$ is not merely a multiset of formulas, but a \emph{bunch}: a tree in which leaves are formulae and nodes are connected with either $\csemic$ or $\ccomma$ (signifying connecting the resources using $\wedge$ and $\ast$, respectively).
For example, a bunch might be $((a \wedge b)\csemic c)\ccomma (d\csemic e)$.
Due to this nested structure of bunches, the left rules in the BI sequent calculus can apply deep inside bunches.
For example, an instance of the left rule for $\wedge$, specialized to the bunch above, is
\begin{mathpar}
  \infer
  {((a\csemic b)\csemic c)\ccomma (d\csemic e) \proves \varphi}
  {((a \wedge b)\csemic c)\ccomma (d\csemic e) \proves \varphi}
\end{mathpar}
That is, $a \wedge b$ got ``destructed'' into $a \csemic b$ in the context ${([-]\csemic c)\ccomma (d\csemic e)}$, where $[-]$ signifies a hole that can be filled.

BI treats separating conjunction $\ast$ (and, hence, $\ccomma$) as a substructural connective, that does not admit contraction and weakening (i.e.~neither $a \proves a \ast a$ nor $a \ast b \proves a$ hold), but it retains the usual structural rules for intuitionistic conjunction $\wedge$ (and, hence, $\csemic)$.
In the sequent calculus, the corresponding structural rules can as well be applied deeply inside bunches.
For example, an instance of a contraction rule might look like this:
\begin{mathpar}
\infer
{\big((a\ccomma b) \csemic (a\ccomma b)\big) \ccomma c \proves \varphi}
{(a\ccomma b) \ccomma c \proves \varphi}
\end{mathpar}
Here we contract the bunch $(a\ccomma b)$ inside the context $[-] \ccomma c$.
In BI we have to permit contraction on arbitrary bunches, whereas in intuitionistic logic contraction on individual formulas is sufficient.

As usual, BI includes a \emph{cut rule}, which formalizes the informal process of applying an intermediate lemma in a proof.
Similar to the other rules, the cut rule can be applied on a formula deeply nested inside a bunch:
\begin{mathpar}
  \infer
  {\Delta' \proves \psi \and \Delta(\psi) \proves \varphi}
  {\Delta(\Delta') \proves \varphi}
\end{mathpar}
where $\Delta(-)$ is an arbitrary bunch with a hole.

In this paper we study the \emph{cut elimination} property for BI.
That is, every sequent that has a proof in BI involves the cut rule also has a proof that is cut-free (i.e. does not use of the cut rule).
From a theoretical point of view, cut elimination can be used to show important meta-theoretical properties (subformula property, consistency, conservativity).
From a more practical standpoint, cut elimination is an important ingredient in proof search.

\paragraph{Why formalize cut elimination?}
Cut elimination is a staple in metatheory of logics.
Because of that, the question of cut elimination is often one of the first to be raised, whenever a new logic or a new sequent calculus is proposed.
It is then common to prove cut elimination directly, by providing a recursive procedure on derivation trees, potentially using additional measure(s) to prove that this procedure terminates.

Proofs organized along those lines are repetitive, consist of many sub-cases, and include many implicit details (e.g. about the structure of the contexts).
As a result, it is not uncommon to see proofs that are ``analogous'' to known correct proofs of cut elimination for related systems, or proofs that only discuss a couple of cases that are considered illustrative, with the bulk of the proof being left as a (rarely completed) exercise for the reader.

Unfortunately, due to the interplay and complexity of all the details, such informal proofs can be quite risky.
In the case of BI, the deep nested structure of bunches and explicit structural rules contribute to the complexity and the level of details.
For example, a proof of cut elimination for BI given in~\cite[Chapter 6]{pym:2002} had a gap, that was later fixed in~\cite{arisaka.qin:2012}.
The issue seems to arise from the treatment of the contraction rule.
In presence of explicit contraction a naive approach of pushing each instance of the cut rule up along the derivation tree does not necessarily work.
In order to resolve this, the cut rule should be generalized to the \emph{multicut} rule, combining contraction and cut together.
Then cut elimination is generalized to multicut elimination, offering a stronger induction hypothesis that can be applied to subproofs.
Unfortunately, this generalization was originally done in a way that only works for some of the cases.
See~\cite{arisaka.qin:2012} for more details.\footnote{It is possible to avoid the multicut generalization by using more fine-grained measure functions, see~\cite{borisavljevic.etal:2000} for the case of intuitionistic logic. As another alternative, Brotherston~\cite{brotherston:2012} gave a proof of cut elimination for BI by going through a displayed calculus.}

This is not the only instance of erroneous proofs of cut elimination slipping in.
Several sequent calculus formulations for bi-intuitionistic logic were wrongly believed to enjoy cut elimination.
These mistakes were later fixed in~\cite{pinto.uustalu:2009}.
Other instances include an incorrect proof of cut elimination for full intuitionistic linear logic, fixed in~\cite{depaiva.brauner:1996,bierman:1996};
an incorrect proof of cut elimination for nested sequent systems for modal logic~\cite{brunnler.strassburger:2009}, fixed in~\cite{marin.strassburger:2014}.
While not incorrect in itself, cut elimination for a formulation of the provability logic GL by Sambin and Valentini~\cite{sambin.valentini:1982} with explicit structural rules was subject of some controversy until it was resolved in~\cite{gore.ramanayake:2012}.

\paragraph{Semantic cut elimination.}
To counterbalance informal pen-and-paper proofs of cut elimination for BI, we provide a fully formalized proof in the Coq proof assistant.
However, instead of trying to formalize an intricate Gentzen-style process, as in \cite{arisaka.qin:2012}, we approach cut elimination using the ideas of algebraic proof theory: a research area aimed at making tight connections between structural proof theory and algebraic semantics of logics.
In our proof we adapt the methods of algebraic semantic cut elimination for linear {logic~\cite{okada:99,okada:02}}, in which cut elimination is obtained by constructing a special model for linear logic that is universal w.r.t. cut-free provability.
We believe that this approach to cut elimination is more amendable to formalization and extension, than a direct Gentzen-style proof.

Semantic cut elimination for BI was first developed by Galatos and Jipsen~\cite{galatos.jipsen:2017}, building on their work on residuated frames~\cite{galatos.jipsen:2013}.
Their approach is quite general, and the proof makes heavy use of intermediate structures (the aforementioned residuated frames), which lie in between sequent calculus and algebraic semantics.
By contrast, the proof presented here  only involves the ``syntax'' (sequent calculus), and the ``semantics'' (algebraic models) parts.
This leaves us with fewer structures to consider in the formalization.

To demonstrate the modularity of our proof, we extend it to cover two different types of extensions of BI.
Firstly, we consider BI extended with a particular class of structural rules (\emph{simple structural rules}), which cover weakening and contraction (both for $\ccomma$ and $\csemic$), as well as many other kinds of structural rules.
Secondly, we consider BI extended with an S4-like $\Box$ modality.
In both cases we show that we do not have to make a lot of modifications to the proof, and the modifications that we do have to make are, in a way, systematic.

\subsection{Contributions and outline}
The main contributions of this paper are as follows.
We present an algebraic proof of cut elimination for BI.
Our proof can be seen as a simplification of the Galatos and Jipsen's method~\cite{galatos.jipsen:2017}, without the framework of residuated frames.
We demonstrate the modularity of our approach by extending it to cut elimination of BI with an S4-like modality (modalities were not previously considered in the framework of residuated frames).
We formalize the results in the Coq proof assistant, which is to our knowledge the first published formalization of cut elimination for BI.

The remained of the paper is structured as follows.
In \Cref{sec:main_idea} we present the main idea behind semantic proofs of cut elimination.
In \Cref{{sec:seqcalc}} and \Cref{sec:alg_sem} we recall the sequence calculus for BI and its (standard) algebraic semantics via BI algebras.
In \Cref{sec:moore} we consider when a closure operator on a BI algebra induces a subalgebra itself.
We then apply this construction in \Cref{sec:cutelim_thm} to obtain a ``universal'' model for cut-free provability, and use it to prove cut elimination.
In \Cref{{sec:simple_ext}} we extend the proof of cut elimination to all possible extensions of BI with a particular class of structural rules.
In \Cref{sec:modal_ext} we extend the proof to account for an S4-like modality.
We discuss our formalization efforts in \Cref{sec:formalization}.
We discuss related work in \Cref{sec:related_work} and conclude in \Cref{sec:conclusion}.

\subsection{Formalization}
The formalization is available online at:
\begin{center}
  \url{https://github.com/co-dan/BI-cutelim}.
\end{center}
In this paper we specifically refer to the version with git hash 
\href{https://github.com/co-dan/BI-cutelim/tree/\longhash}{\shorthash}, permanently available under DOI \href{https://doi.org/10.5281/zenodo.5770478}{10.5281/zenodo.5770478}.
Throughout the paper, identifiers in monospaced font (\href{https://github.com/co-dan/BI-cutelim}{\nolinkcoqident{like this}}) accompany statements and proposition.
They indicate the names of the statements in the Coq formalization and link to the corresponding place in the online documentation.
For example, the link \coqident{seqcalc}{proves} points to the inductive definition of the BI sequent calculus.

\section{Semantic cut elimination}
\label{sec:main_idea}
In this section we explain some of the ideas and intuitions behind a semantic proof of cut elimination in a semi-formal way, before diving straight into the complexities of BI.
The starting point is that there is a class of algebras in which we can interpret logic.
The main idea is to find a particular algebra $\mathcal{C}$, in which we can interpret the sequent calculus, and which has a property that if $\Sem{\psi} \leq \Sem{\varphi}$ in $\mathcal{C}$, then $\psi \proves \varphi$ is derivable without applications of the \ruleref{cut} rule.
In this case, we say that $\mathcal{C}$ is a ``universal'' algebra for cut-free provability.
Then, cut elimination can be obtained by the (sound) interpretation of sequent calculus into $\mathcal{C}$.

Finding such a ``universal'' algebra is reminiscent of proving \emph{completeness} of a logic w.r.t. a class of algebras.
In the case of completeness, we construct a ``universal'' algebra $\mathcal{L}$ such that $\Sem{\psi} \leq \Sem{\varphi}$ implies derivability of $\psi \proves \varphi$.
This Lindenbaum-Tarski algebra $\mathcal{L}$ is usually defined to be the collection of equivalence classes of formulas modulo interprovability:
\[
  [\varphi] \eqdef \{ \psi \mid ( \psi \proves \varphi)~\wedge~(\varphi \proves \psi)\}
\]
And the ordering $\leq$ on $\mathcal{L}$ is induced by provability:
\[
  [\varphi] \leq [\psi] \iff \varphi \proves \psi.
\]
Provability does not depend on the representative of the equivalence class, and so we get a poset $\mathcal{L}$.
The logical operators are interpreted in $\mathcal{L}$ in such a way that $\Sem{\varphi} = [\varphi]$.
The argument for completeness then goes as follows: suppose that $\Sem{\varphi} \leq \Sem{\psi}$ in all the possible algebras; then, in particular that inequality holds in $\mathcal{L}$, which amounts to $\varphi \proves \psi$.
Thus, any valid sequent is derivable.

It is precisely the connection between provability and the order on the algebra that makes this model useful.
We can imagine a reformulation of the above model in terms of cut-free provability in sequent calculus:
\[
  [\varphi] \leq [\psi] \iff \varphi \provesCF \psi.
\]
This adaptation, however, does not work.
In order to prove that the ordering $\leq$ is transitive, we need to show
\begin{mathpar}
\infer{
  \varphi \provesCF \psi
  \and \psi \provesCF \chi
}
{\varphi \provesCF \chi}
\end{mathpar}
which amounts to showing that \ruleref{cut} is admissible in the cut-free fragment.
We seem to be back at square one.

To fix this, instead of interpreting formulas as sets of equivalent formulas (which is what equivalence classes can be seen as),
we would like to interpret formulas as sets of \emph{contexts} which prove the formula:
\[
  \outI{\varphi} \eqdef \{ \Delta \mid \Delta \provesCF \varphi \}.
\]
Then, inclusion of sets is a good candidate for the ordering, because $\psi \in \outI{\psi}$ and, hence, $\outI{\psi} \subseteq \outI{\varphi}$ implies $\psi \provesCF \varphi$.

But how do we interpret logical connectives?
We can interpret $\top$ as the set of all contexts; then, clearly $\top = \outI{\TRUE}$.
However, we cannot pick the empty set as an interpretation of $\bot$: the set $\outI{\FALSE}$ is non-empty, as it contains at least $\FALSE$ itself.
What we need is to find an interpretation $\Sem{-}$ such that $\Sem{\varphi} \subseteq \Sem{\psi}$ implies $\varphi \provesCF \psi$ (or, equivalently $\varphi \in \outI{\psi}$).
Okada~\cite{okada:99} proposed a sufficient condition for such an interpretation: for any formula $\varphi$, $\varphi \in \Sem{\varphi}$ and $\Sem{\varphi} \in \outI{\varphi}$.
Then, the desired property on the interpretation follows via a chain of inclusions:
\[
  \varphi \in \Sem{\varphi} \subseteq \Sem{\psi} \subseteq \outI{\psi}.
\]

Note that the set of all contexts does not satisfy this condition: as we have seen, the empty set a counter-example.
It is the least element w.r.t~set inclusion, but $\FALSE \notin \emptyset$, so we cannot set $\Sem{\FALSE} = \emptyset$.
This suggests that, instead of considering arbitrary sets of contexts, we need to refine the powerset algebra somehow.
A good starting point would be to consider the carrier of the algebra containing just the sets of the form $\outI{\varphi}$,
i.e.~$\mathcal{C} = \{ \outI{\varphi} \mid \varphi \in \Frml \}$ (by analogy with the Lindenbaum-Tarski algebra, which consists only of elements of the form $[\varphi]$).
We can then interpret bottom as $\bot = \outI{\FALSE}$, and it indeed will be the least element in the algebra.

Looking at other connectives, we cannot interpret disjunction as set-theoretic union, because the union $\outI{\varphi} \cup \outI{\psi}$ cannot always be written as $\outI{\chi}$, for some formula $\chi$.
That is, we cannot actually show that $\mathcal{C}$, as given above, is closed under unions, so $\cup$ is not a well-defined operation on $\mathcal{C}$.

How should we then interpret disjunction if not as the union of sets?
If we cannot use the set union, we will use the ``next best thing'': the smallest set in $\mathcal{C}$ that actually contains the union.
Formally, we set:
$$X \vee Y = \bigcap\{Z \in \mathcal{C} \mid X \cup Y \subseteq Z \}.$$
This definition is still not without issues:
for this operation to be defined, we need to ensure that $\mathcal{C}$ is closed under arbitrary intersections.
It turns out that we can achieve this by modifying the carrier of $\mathcal{C}$ and ``baking in'' the closedness under arbitrary intersections.
Such a construction is obtained in a generic way as a subalgebra of the powerset algebra generated by a particular \emph{closure operator}, as we will see in \Cref{sec:moore,sec:cutelim_thm}.

In the remainder of the paper we develop this construction in details.
But first, to make the matters concrete, we recall the BI sequent calculus and properties of its cut-free fragment (\Cref{sec:seqcalc}), and the algebraic semantics for BI (\Cref{{sec:alg_sem}}).


\section{Sequent calculus for BI}
\label{sec:seqcalc}
\begin{figure*}[t]
  \begin{mdframed}
    \raggedright
    \textbf{Equivalence of bunches}
    \begin{mathpar}
      \axiom{\Delta_1 \ccomma \Delta_2 \equiv \Delta_2 \ccomma \Delta_1}
      \and
      \axiom{\Delta_1 \csemic \Delta_2 \equiv \Delta_2 \csemic \Delta_1}
      \and
      \axiom{\Delta_1 \ccomma (\Delta_2 \ccomma \Delta_3) \equiv (\Delta_1 \ccomma \Delta_2) \ccomma \Delta_3}
      \and
      \axiom{\Delta_1 \csemic (\Delta_2 \csemic \Delta_3) \equiv (\Delta_1 \csemic \Delta_2) \csemic \Delta_3}
      \and
      \\
      \axiom{\Delta \ccomma \empM \equiv \Delta}
      \and
      \axiom{\Delta \csemic \empA \equiv \Delta}
      \and
      \infer
      {\Delta \equiv \Delta'}
      {\Gamma(\Delta) \equiv \Gamma(\Delta')}
    \end{mathpar}
    \smallskip
    \textbf{Structural rules}
    \begin{mathpar}
      \inferH{ax}
      {a \in \Atom}
      {a \proves a}
      \and
      \inferH{equiv}
      {\Delta' \proves \varphi{} \and \Delta \equiv \Delta'}
      {\Delta \proves \varphi}
      \and
      \inferhref{W$\csemic$}{W;}
      {\Delta(\Delta_1) \proves \varphi}
      {\Delta(\Delta_1 \csemic \Delta_2) \proves \varphi}
      \and
      \inferhref{C$\csemic$}{C;}
      {\Delta(\Delta_1 \csemic \Delta_1) \proves \varphi}
      {\Delta(\Delta_1) \proves \varphi}
      \and
      \inferH{cut}
      {\Delta' \proves A
        \and \Delta(A) \proves B}
      {\Delta(\Delta') \proves B}
    \end{mathpar}
    \smallskip

    \raggedright
    \textbf{Multiplicatives}
    \begin{mathpar}
      \axiomhref{$\EMP$R}{empR}
      {\empM \proves \EMP}
      \and
      \inferhref{$\EMP$L}{empL}
      {\Delta(\empM) \proves \varphi}
      {\Delta(\EMP) \proves \varphi}
      \and
      \inferhref{$\ast$R}{sepR}
      {\Delta_1 \proves \varphi \and \Delta_2 \proves \psi}
      {\Delta_1\ccomma \Delta_2 \proves \varphi \ast \psi}
      \and
      \inferhref{$\ast$L}{sepL}
      {\Delta(\varphi\ccomma \psi) \proves \chi}
      {\Delta(\varphi \ast \psi) \proves \chi}
      \and
      \inferhref{$\wand$R}{wandR}
      {\Delta\ccomma \varphi \proves \psi}
      {\Delta \proves \varphi \wand \psi}
      \and
      \inferhref{$\wand$L}{wandL}
      {\Delta_1 \proves \varphi \and
        \Delta(\Delta_2\ccomma \psi) \proves \chi}
      {\Delta(\Delta_1\ccomma \Delta_2\ccomma \varphi \wand \psi) \proves \chi}
    \end{mathpar}
    \smallskip

    \raggedright
    \textbf{Additives}
    \begin{mathpar}
      \axiomhref{$\TRUE$R}{trueR}
      {\empA \proves \TRUE}
      \and
      \inferhref{$\TRUE$L}{trueL}
      {\Delta(\empA) \proves \varphi}
      {\Delta(\TRUE) \proves \varphi}
      \and
      \inferhref{$\wedge$R}{andR}
      {\Delta_1 \proves \varphi \and \Delta_2 \proves \psi}
      {\Delta_1 \csemic \Delta_2 \proves \varphi \wedge \psi}
      \and
      \inferhref{$\wedge$L}{andL}
      {\Delta(\varphi\csemic \psi) \proves \chi}
      {\Delta(\varphi \wedge \psi) \proves \chi}
      \and
      \inferhref{$\to$R}{implR}
      {\Delta \csemic \varphi \proves \psi}
      {\Delta \proves \varphi \to \psi}
      \and
      \inferhref{$\to$L}{implL}
      {\Delta_1 \proves \varphi \and
        \Delta(\Delta_2\csemic \psi) \proves \chi}
      {\Delta(\Delta_1\csemic \Delta_2\csemic \varphi \to \psi) \proves \chi}
      \and
      \axiomhref{$\FALSE$L}{falseL}
      {\Delta(\FALSE) \proves \varphi}
      \and
      \inferhref{$\vee$R1}{disjR1}
      {\Delta \proves \varphi}
      {\Delta \proves \varphi \vee \psi}
      \and
      \inferhref{$\vee$R2}{disjR2}
      {\Delta \proves \psi}
      {\Delta \proves \varphi \vee \psi}
      \and
      \inferhref{$\vee$L}{disjL}
      {\Delta(\varphi) \proves \chi \and \Delta(\psi) \proves \chi}
      {\Delta(\varphi \vee \psi) \proves \chi}
    \end{mathpar}
  \end{mdframed}
  \caption{BI sequent calculus.}
\label{fig:bi_seqcalc}
\end{figure*}
In this section we briefly recall the sequent calculus formulation of BI \cite{ohearn:pym:99}, and some of the properties of its cut-free fragment.
The formulas of BI are obtained from the following grammar:
\begin{align*}
\varphi, \psi & \bnfdef{}
	\TRUE \ALT \FALSE \ALT 
        \varphi \wedge \psi \ALT
        \varphi \vee \psi \ALT
        \varphi \to \psi \\ & \ALT
        \EMP \ALT
	\varphi * \psi \ALT
	\varphi \wand \psi \ALT a \qquad\quad (a \in \textdom{Atom})
\end{align*}
BI extends intuitionistic propositional logic with separating conjunction ($\ast$), magic wand ($\wand$, adjoint to separating conjunction), and the empty proposition ($\EMP$, unit for separating conjunction).
We also include atomic propositions drawn from a fixed set $\textdom{Atom}$.

The sequent calculus for BI 
is given in \Cref{fig:bi_seqcalc}.
It operates on the sequents of the form $\Delta \proves \varphi$, where $\varphi$ is a formula and $\Delta$ is a \emph{bunch} -- a tree composed of binary nodes labeled with $\ccomma$ and $\csemic$, and leaves being either formulas or empty bunches $\empM$ and $\empA$.
Morally, we view bunches as equivalence classes of such trees modulo commutative monoid laws for $(\ccomma, \empM)$, and $(\csemic, \empA)$.
These are given using structural congruence $\equiv$, the rules for which are also given in \Cref{fig:bi_seqcalc}.
We could have defined provability on such equivalence classes, but we opt for using explicit context conversions using \ruleref{equiv}.

Most of the structural rules and the left rules can be applied to formulas that occur nested inside some bunch with a hole $\Delta(-)$.
We refer to such bunches with holes as \emph{bunched contexts}.
For example, in the application of the rule \ruleref{andL} below we use the bunched context $(p \ccomma [-])$:
\begin{mathpar}
\inferrule*[left=\ruleref{andL}]
{p \ccomma (p \csemic q) \proves p \ast q}
{p \ccomma (p \wedge q) \proves p \ast q.}
\end{mathpar}

\subsection{Cut-free provability}
Let us write $\Delta \provesCF \varphi$ if $\Delta \proves \varphi$ is derivable \emph{without} the \ruleref{cut} rule.
In the rest of this section we prove invertibility of several rules in the cut-free fragment of BI.
Those derived rules will be useful to us when constructing the algebraic model in \Cref{sec:cutelim_thm}.

The first observation about the sequent calculus, is that we have formulated the ``axiom'' rule $\varphi \proves \varphi$ only for atomic formulas $a \in \Atom$.
This will significantly simplify some of the proofs (for example, \Cref{lem:sepL-inv}), but does not limit the expressivity of the system, as witness by the following lemma.
\begin{proposition}[Identity expansion, \coqmod{seqcalc}{Seqcalc}{seqcalc_id_ext}]
  \label{lem:id_expansion}
  For every formula $\varphi$ we can derive a sequent
  $\varphi \provesCF \varphi$.
\end{proposition}
\begin{proof}
  By induction on the structure of $\varphi$.
\end{proof}

For the construction presented in this paper we need to show that a number of rules are invertible in the cut-free sequent calculus.
Specifically, we need to show that \ruleref{wandR}, \ruleref{implR}, \ruleref{sepL}, \ruleref{andL}, \ruleref{empL}, and \ruleref{trueL} are invertible.
\begin{lemma}[\coqmod{seqcalc_height}{SeqcalcHeight}{wand_r_inv} and \coqmod{seqcalc_height}{SeqcalcHeight}{impl_r_inv}]
  \label{lem:wand_inv_adm}
  The following rules are admissible:
  \begin{mathpar}
    \inferhref{$\wand$R-inv}{wandR-inv}
  {\Delta \provesCF \varphi \wand \psi}
  {\Delta \ccomma \varphi \provesCF \psi}  
  \and
    \inferhref{$\to$R-inv}{implR-inv}
  {\Delta \provesCF \varphi \to \psi}
  {\Delta\csemic \varphi \provesCF \psi}  
  \end{mathpar}
\end{lemma}
\begin{proof}
  By induction on the derivations $\Delta \provesCF \varphi \wand \psi$
and $\Delta \provesCF \varphi \to \psi$.
\end{proof}
At the end of the day, the proof of \Cref{lem:wand_inv_adm} by induction on derivations is not very complicated, because the form of the context on the left-hand side of the sequent is relatively simple.
It is easy to show that the left rules can commute with \ruleref{wandR} and \ruleref{implR}.
By contrast, showing that the rules \ruleref{sepL} and \ruleref{andL} are invertible is more involved for several reasons.

First of all, just like for other sequent calculi with explicit contraction, structural induction on the proof is not strong enough.
Consider the following derivation of ${\varphi \ast \psi \provesCF \chi}$:
\begin{mathpar}
\infer*
{\varphi \ast \psi \csemic \varphi \ast \psi \provesCF \chi}
{\varphi \ast \psi \provesCF \chi}
\end{mathpar}
Since $\varphi \ast \psi$ occurs twice in the premise, we need to apply the induction hypothesis twice.
From the first application of the induction hypothesis we get a proof
\[
\varphi \ast \psi \csemic (\varphi \ccomma \psi) \provesCF \chi,
\]
but this proof is not a strict subderivation of the original derivation.
Therefore, we cannot use the induction hypothesis second time to obtain a proof of
$(\varphi \ccomma \psi) \csemic (\varphi \ccomma \psi) \provesCF \chi$.

In order to circumvent this, we do induction on the \emph{height} of the derivation, strengthening the statement to:
\begin{lemma}[\coqmod{seqcalc_height}{SeqcalcHeight}{sep_l_inv}]
\label{lem:sepL-inv}
If there is a derivation of $$\Delta(\varphi \ast \psi) \provesCF \chi$$
  with height $n$, then there is a derivation of $$\Delta(\varphi\ccomma \psi) \provesCF \chi$$ with height strictly less than $n$.
\end{lemma}
Note that this lemma would be false if we would have included an axiom rule for arbitrary formulas: there would be a proof $\varphi \ast \psi \provesCF \varphi \ast \psi$  of height 0, but the smallest proof of $\varphi\ccomma \psi \provesCF \varphi \ast \psi$ is of height 1.
That is why we have restricted the axiom rule to atomic formulas, and got the general form of the axiom rule as a derived statement (\Cref{lem:id_expansion}).

Similarly, by induction on the derivation height, we show that \ruleref{empL} and \ruleref{trueL} are invertible.
We only care about the derivation height for the purposes of induction, so we summarize the results on invertible rules in the following lemma.
\begin{lemma}
  \label{lem:inv_rules}
  The following rules are admissible:
  \begin{mathpar}
  \inferhref{$\wedge$L-inv}{andL-inv}
  {\Delta(\varphi \wedge \psi) \provesCF \chi}
  {\Delta(\varphi \csemic \psi) \provesCF \chi}  
  \and
  \inferhref{$\ast$L-inv}{sepL-inv}
  {\Delta(\varphi \ast \psi) \provesCF \chi}
  {\Delta(\varphi \ccomma \psi) \provesCF \chi}  
  \and\\
  \inferhref{$\top$L-inv}{trueL-inv}
  {\Delta(\TRUE) \provesCF \chi}
  {\Delta(\empA) \provesCF \chi}  
  \and
  \inferhref{$\EMP$L-inv}{empL-inv}
  {\Delta(\EMP) \provesCF \chi}
  {\Delta(\empM) \provesCF \chi}  
  \end{mathpar}
\end{lemma}
Let us write $\frmlI{\Delta}$ for interpretation of $\Delta$ as a formula: we substitute every occurrence of $\ccomma$ in $\Delta$ with $\ast$, every occurrence of $\empM$ with $\EMP$, and similarly for the additive connectives.
Clearly, there is a derivation from $\Delta \provesCF \chi$ to $\frmlI{\Delta} \provesCF \chi$, by repeated application of \ruleref{sepL} and \ruleref{andL}.
For the other direction we have the following.
\begin{corollary}[\coqmod{seqcalc_height}{SeqcalcHeight}{collapse_l_inv}]
\label{lem:comprehension_adm}
The following rule is admissible:
\begin{mathpar}
\infer
{\Delta'(\frmlI{\Delta}) \provesCF \chi}
{\Delta'(\Delta) \provesCF \chi}  
\end{mathpar}
\end{corollary}
\begin{proof}
  By induction on $\Delta$, using \Cref{{lem:inv_rules}}.
\end{proof}


\section{Algebraic semantics for BI}
\label{sec:alg_sem}
We interpret the BI sequent calculus in the algebraic structures known as BI algebras, which are bounded Heyting algebras with a compatible residuated monoidal structure.

\begin{definition}
  \begin{sloppypar}
    A BI algebra $\mathcal{B}$ is a tuple \sloppy${(B, \bot, \top, \wedge, \vee, \to, \EMP, \ast, \wand)}$ where
  \end{sloppypar}
  \begin{itemize}
  \item $(B, \bot, \top, \wedge, \vee, \to)$ is a bounded Heyting algebra, i.e. a bounded distributive lattice with the Heyting implication satisfying
    \[
      a \wedge b \leq c \iff a \leq b \to c
    \]
  \item ${\ast} : {\cal B} \times {\cal B} \to {\cal B}$ is a monotone commutative and associative function;
  \item $\EMP : {\cal B}$ is a unit element for $\ast$;
  \item ${\wand} : {\cal B} \times {\cal B} \to {\cal B}$ is a binary operation satisfying
    \[
      a \ast b \leq c \iff a \leq b \wand c
    \]
  \end{itemize}
\end{definition}

\begin{definition}
  Let ${\cal B}$ be an arbitrary BI algebra.
  Given an interpretation $i : \Atom \to {\cal B}$ of atomic propositions, 
  we interpret formulas of BI in ${\cal B}$ in the usual tautological way:
\begin{align*}
\Sem{\EMP} &{} = \EMP & \Sem{\TRUE} &{} = \top \\
\Sem{\varphi \ast \psi} &{} = \Sem{\varphi} \ast \Sem{\psi} &
\Sem{\varphi \wedge \psi} &{} = \Sem{\varphi} \wedge \Sem{\psi} \\
\Sem{\varphi \wand \psi} &{} = \Sem{\varphi} \wand \Sem{\psi} &
\Sem{\varphi \to \psi} &{} = \Sem{\varphi} \to \Sem{\psi} \\
\Sem{\varphi \vee \psi} &{} = \Sem{\varphi}\vee \Sem{\psi} &
\Sem{\FALSE} &{} = \bot\\
\Sem{a} &{} = i(a)
\end{align*}  
\end{definition}
\begin{theorem}[Soundness, \coqmod{seqcalc}{Seqcalc}{seq_interp_sound}]
  \label{thm:soundness}
  If $\Delta \proves \varphi$ is derivable, then $\Sem{\frmlI{\Delta}} \leq \Sem{\varphi}$ holds in any BI algebra.
\end{theorem}
\begin{proof}
  By induction on the derivation.
\end{proof}

\subsection{BI algebras from monoids}
\label{sec:bi_alg_from_pcm}
In practice, a lot of BI algebras arise as predicates over a partial commutative monoid.
Let $(M, \cdot, e)$ be a partially commutative monoid; we write $x \cdot y = \bot$ if composition of $x$ and $y$ is undefined.
Then the powerset $\pset{M}$ is a (complete) Heyting algebra, and it forms a BI algebra ${(\pset{M}, \emptyset, M, \cap, \cup, \to, \mathbf{0}, \opBullet, \wandBullet)}$ with the following operators:
\begin{align*}
  \mathbf{0} & \eqdef \{ e \} \\
  X \opBullet Y & \eqdef \{ x \mult y \mid x \in X, y \in Y, x \mult y \neq \bot \}\\
  X \wandBullet Y & \eqdef \{ z \mid \All x \in X. z \mult x \neq \bot \implies z \mult x \in Y \}.
\end{align*}

\paragraph{BI algebra from the monoid of contexts.}
Let us write $\Bunch$ for the set of bunches, modulo the equivalence $\equiv$ (from \Cref{fig:bi_seqcalc}).
We can endow the set $\Bunch$ of bunches with the structure of a monoid.
Composition of two contexts $\Delta$ and $\Delta'$ is just putting them next to each other using $\ccomma$:
\[
  \Delta \cdot \Delta' \eqdef (\Delta\ccomma \Delta')
\]
then, up to equivalence of bunches, $\empM$ is the unit element.
Using the powerset construction we get a BI algebra $\pset{\Bunch}$.

This model is very much ``freely generated'' from syntax, but it is not very useful, as it does not involve any notion of provability (only equivalence of contexts).
In this next sections we are going to refine this model, in order to obtain a submodel which can be used to prove completeness and cut-elimination.


\section{Moore closures on BI algebras}
\label{sec:moore}
For cut elimination, we will be interested in subalgebras of the powerset algebra $\pset{M}$ for some partial commutative monoid $M$; specifically subalgebras arising from a particular closure operator.
For the rest of this section we fix a partial commutative monoid $M$.

\begin{definition}
  A Moore collection is a family of sets $\mathcal{C} \subseteq \pset{M}$ that is closed under arbitrary intersections:
  \[
    (\All i \in I. A_i \in \mathcal{C}) \implies \bigcap_{i \in I} A_i \in \mathcal{C}.
  \]

  If $X \in \mathcal{C}$ we say that $X$ is \emph{closed}.
\end{definition}

Alternatively, a Moore collection can be given in terms of a closure operator $\cl{-}$ satisfying the following conditions:
\begin{itemize}
\item $X \subseteq \cl{X}$;
\item monotonicity: $X \subseteq Y \implies \cl{X} \subseteq \cl{Y}$;
\item idempotence: $\cl{\cl{X}} = \cl{X}$.
\end{itemize}
Given a Moore collection $\mathcal{C}$ we define the associated closure operator as $\cl{X} = \bigcap \{ Y \in \mathcal{C} \mid X \subseteq Y \}$.
In the other direction, given a closure operator we define $\cl{-}$-closed sets as $\mathcal{C} = \{ X \mid \cl{X} = X \}$.

Some basic theory behind posets with such a closure operator is given in~\cite{everett:1944}.
Here, we recall only the results that we will be needing.
First of all, we are going to use the following rule often.
\begin{lemma}[\coqident{algebra.from_closure}{cl_adj}]
The closure operator satisfies the following adjunction rule:
\begin{mathpar}
\inferB
{X \subseteq Y \quad \mbox{ in $\pset{M}$}}
{\cl{X} \subseteq Y \quad \mbox{ in ${\cal C}$}}
\end{mathpar}
for a closed set $Y$.
\end{lemma}

Since $\mathcal{C}$ is closed under intersections, $X \cap Y$ is a meet of two closed sets $X$ and $Y$.
However, given two closed sets, their union $X \cup Y$ is not always closed.
Instead, we interpret join as $\cl{X \cup Y}$.
\begin{proposition}
The collection $\mathcal{C}$ is a complete bounded lattice: the least upper bound is given by $\bigvee_{i \in I} X_i = \cl{\bigcup_{i \in I} X_i}$.
In particular, the bottom element of $\mathcal{C}$ is $\cl{\emptyset}$.
\end{proposition}
It is not necessarily the case that $\mathcal{C}$ has Heyting implication, but if it does, then we can describe it in terms of the implication on $\pset{M}$ and a dual of the closure operator.
\begin{proposition}[\coqident{algebra.from_closure}{impl_from_int}]
  For a set $X \in \pset{M}$, we write $\mathsf{int}(X)$ for the largest closed set contained in $X$:
  \[
    \mathsf{int}(X) = \bigvee \{ {Y \in \mathcal{C}} | Y \subseteq X \}.
  \]
  Then for closed sets $X$ and $Y$,
  \[
    X \to Y = \mathsf{int}(X \supset Y)
  \]
  where $X \supset Y$ denotes implication in $\pset{M}$.
\end{proposition}
\begin{proof}
We reason as follows:
\begin{align*}
\mathsf{int}(X \supset Y) &{}=  \bigvee \{ {Z \in \mathcal{C}} | Z \subseteq X \supset Y \} \\
&{} = \bigvee \{ {Z \in \mathcal{C}} | Z \cap X \subseteq Y \} \\
&{} = \bigvee \{ {Z \in \mathcal{C}} | Z \subseteq X \to Y \} = X \to Y
\end{align*}
\end{proof}
In light of the previous propositions, we can see that some Heyting algebra structure on $\mathcal{C}$ arises from the same operations on $\pset{M}$.
Can we similarly lift the BI operations?
Let us denote the residuated monoidal structure (defined as in \Cref{{sec:bi_alg_from_pcm}}) on $\pset{M}$ as $(\mathbf{0}, \opBullet, \wandBullet)$.
In the rest of this section we describe how to lift this structure to $\mathcal{C}$.

\subsection{BI algebra structure on closed sets}
\label{sec:bi_algebra_moore}
A sufficient condition for $\mathcal{C}$ to be a BI algebra is the following.%
\begin{definition}
  We say that the closure operator is \emph{strong} if for any $X$ and $Y$
  \[
    \cl{X} \opBullet Y \subseteq \cl{X \bullet Y}
  \]
\end{definition}
If $\cl{-}$ is strong, then we define the BI operators on $\mathcal{C}$ as follows:
\begin{align*}
\EMP &{}= \cl{\mathbf{0}} \\
X \ast Y &{} = \cl{X \opBullet Y} \\
X \wand Y &{} = \cl{X \wandBullet Y}
\end{align*}
We shall verify that with these connectives $\mathcal{C}$ is a BI algebra.

\begin{proposition}[\coqident{algebra.from_closure}{wand_intro_r}, \coqident{algebra.from_closure}{wand_elim_l'}]
  \label{prop:wand_adjunction}
  There is an adjunction between $\ast$ and $\wand$:
  $$X \ast Y \subseteq Z \quad \iff \quad X \subseteq Y \wand Z.$$
\end{proposition}
\begin{proof}
We reason as follows.
\begin{align*}
&X \ast Y \subseteq Z {} & \text{(def. of $\ast$)}\\
&\iff \cl{X \opBullet Y} \subseteq Z & \text{($Z$ is closed)}\\
&\iff X \opBullet Y \subseteq Z & \text{(adjunction)}\\
&\iff X \subseteq Y \wandBullet Z  \\
&\,\;\implies X \subseteq \cl{Y \wandBullet Z} & \text{(def. of $\wand$)}\\
&\iff X \subseteq Y \wand Z.
\end{align*}
On the other hand,
\begin{align*}
&X \subseteq \cl{Y \wandBullet Z} & \text{(monotonicity of $\opBullet$)}\\
&\,\;\implies X \opBullet Y \subseteq \cl{Y \wandBullet Z} \opBullet Y & \text{(strength of $\cl{-}$)}\\
&\,\;\implies X \opBullet Y \subseteq \cl{(Y \wandBullet Z) \opBullet Y} \\
&\,\;\implies X \opBullet Y \subseteq \cl{Z} = Z \\
&\iff \cl{X \opBullet Y} \subseteq Z.
\end{align*}
\end{proof}

\begin{proposition}[\coqident{algebra.from_closure}{sep_comm'}, \coqident{algebra.from_closure}{sep_assoc'}, and \coqident{algebra.from_closure}{emp_sep_1}, \coqident{algebra.from_closure}{emp_sep_2}]
  $(\mathcal{C}, \ast, \EMP)$ is a commutative monoid.
\end{proposition}
\begin{proof}
  The commutativity of $\ast$ is evident from its definition.
  Let us verify the unit laws:
\begin{align*}
&\EMP \ast X {} = \cl{\cl{\mathbf{0}} \opBullet X} \subseteq
\cl{\cl{\mathbf{0} \opBullet X}} {}= X \\
&X {}= \mathbf{0} \opBullet X \subseteq \cl{\mathbf{0}} \opBullet X \subseteq \cl{\cl{\mathbf{0}} \opBullet X} {} = \EMP \ast X.
\end{align*}
We reason similarly for associativity of $\ast$.
\end{proof}

We can summaries these results in the following theorem.
\begin{theorem}
\label{thm:bi_alg_lifting_moore}
Let $M$ be a PCM, and let  $\cl{-}$  be a strong closure operator on $\pset{M}$, such that $\mathcal{C}$ has Heyting implication.
Then the set $\mathcal{C}$ of closed elements is a BI algebra.
\end{theorem}

Finally, some times it is more convenient to use an alternative condition in place of closure strength:
\begin{proposition}
  \label{prop:bi_lift_strong_closure}
  The closure operator is strong iff $X \wandBullet Y$ is closed whenever $Y$ is closed, i.e.~ $\mathcal{C}$ forms an exponential ideal.
\end{proposition}
\begin{proof}
  Suppose that $\mathcal{C}$ is an exponential ideal w.r.t~ $\wandBullet$.
  Then we reason as follows:
  \begin{mathpar}
  \infer*
  {\cl{X} \opBullet Y \subseteq \cl{X \opBullet Y}}
  {\infer*
    {\cl{X} \subseteq Y \wandBullet \cl{X \opBullet Y}}
    {\infer*
      {X  \subseteq Y \wandBullet \cl{X \opBullet Y} \and \mbox{\emph{(the r.h.s.~is closed)}}}
      {X \opBullet Y \subseteq \cl{X \opBullet Y}}}}
  \end{mathpar}
Hence, $\cl{-}$ is strong.

For the other direction, if $Y$ is closed, then
\begin{align*}
\cl{X \wandBullet Y} \subseteq X \wandBullet Y 
  \iff{}&  \cl{X \wandBullet Y} \opBullet X \subseteq Y\\
  \Longleftarrow{}\,\;& {}\cl{(X \wandBullet Y) \opBullet X} \subseteq Y \\
  \iff{}& (X \wandBullet Y) \opBullet X \subseteq Y
\end{align*}
\end{proof}

\paragraph{A remark on (im)predicativity.}
In practice, we want to start with some collection $\mathcal{B} \subseteq \pset{M}$ of sets, and generate $\mathcal{C}$ freely from arbitrary intersections of elements of $\mathcal{B}$ (think of generating a topology from a closed basis).
Then $\mathcal{C}$ is a Moore collection and the associated closure operator can be given as $\cl{X} = \bigcap \{ Y \in \mathcal{C} \mid X \subseteq Y \}$.
Unfortunately, this definition is impredicative (we define an element of $\mathcal{C}$ by quantifying over elements of $\mathcal{C}$), which, when formalized in type theory, increases the universe level.

That means that we cannot use the closure operator to \emph{define} the set $\mathcal{C}$, i.e. the set $\{ X \mid X = \cl{X} \}$ will have a higher universe level than $\mathcal{C}$.
To circumvent this, we can instead define the closure operator equivalently by quantifying not over all the closed sets, but only over the basic closed sets: $\cl{X} = \bigcap \{ Y \in \mathcal{B} \mid X \subseteq Y\}$.
Then we can define $\mathcal{C}$ to be the set of elements satisfying $X = \cl{X}$.


\section{Cut-elimination via a syntactic model}
\label{sec:cutelim_thm}
In this section we construct a special BI algebra ${\cal C} \subseteq \pset{\Bunch}$ that has the following property: if $\Sem{\varphi} \leq \Sem{\psi}$ holds in ${\cal C}$, then $\varphi \provesCF \psi$.
By composing this with the soundness theorem, we will obtain the cut-elimination result.

\subsection{Principal closed sets}
We are going to construct $\mathcal{C}$ as a particular Moore collection on $\pset{\Bunch}$.
To define when a predicate $X$ is closed (e.g. when $X \in \mathcal{C}$), we start with \emph{principal closed elements}, and generate $\mathcal{C}$ as families of intersections of principal closed sets.

\begin{definition}
  A \emph{principal} closed set is a set of the form:
\[
  \outI{\varphi} = \{ \Delta \mid \Delta \provesCF \varphi \}
\]
for a formula $\varphi$.
\end{definition}

We can then generate closed sets by closing the collection of principal closed sets under arbitrary intersections:
\[
  \cl{X} \eqdef \bigcap\{ \outI{\varphi} \mid X \subseteq \outI{\varphi}\}
  = \bigcap\{ \outI{\varphi} \mid \All \Delta \in X. \Delta \provesCF \varphi \}.
\]
We then define the collection $\mathcal{C}$ of closed sets as
\[
  \mathcal{C} \eqdef \{X \mid X = \cl{X} \}.
\]
Then every element of $\mathcal{C}$ can be written as some intersection $\bigcap_{i \in I}\outI{\varphi_i}$.

Let us briefly describe some useful properties of closed sets:
\begin{proposition}[\coqident{cutelim}{C_inhabited}, \coqident{cutelim}{C_weaken}, \coqident{cutelim}{C_contract}, and \coqident{cutelim}{C_collapse}]
  \label{lem:closed_set_laws}
  Let $X$ be a closed set.
  Then the following holds.
  \begin{enumerate}
  \item $\FALSE \in X$;
  \item $\Delta \in X \implies (\Delta\csemic\Delta') \in X$;
  \item $(\Delta\csemic\Delta) \in X \implies \Delta \in X$;
  \item $\Delta \in X \iff \frmlI{\Delta} \in X$.
  \end{enumerate}
\end{proposition}
\begin{proof}
  For the first point, observe that $\FALSE \provesCF \varphi$, so $\FALSE \in \outI{\varphi}$ for any formula $\varphi$.

  For the second point, let $X$ be $\bigcap_{i \in I}\outI{\varphi_i}$.
  Then, $\Delta \in X \iff \All i \in I. \Delta \provesCF \varphi_i$.
  If $\Delta \in X$, then, using weakening:
  \begin{mathpar}
    \infer
    {\Delta \provesCF \varphi_i}
    {\Delta\csemic \Delta' \provesCF \varphi_i}
  \end{mathpar}
  for any $i \in I$.
  Hence, $\Delta; \Delta' \in X$.

  Similarly for the other two cases, using contraction, and the left rules, and \Cref{lem:comprehension_adm}.
\end{proof}

As an example of a calculation in $\mathcal{C}$, we show the following characterization of meets.
\begin{proposition}[\coqident{cutelim}{C_and_eq}]
\label{prop:meet_in_C}
The following holds in $\mathcal{C}$:
\[
  X \wedge Y = \cl{\{\Delta\csemic \Delta' \mid \Delta \in X, \Delta' \in Y \}}
\]
\end{proposition}
\begin{proof}
  For the inclusion from left to right: suppose that $\Delta \in X \cap Y$.
  Then,
  \begin{align*}
    (\Delta\csemic \Delta) \in{} & \{\Delta\csemic \Delta' \mid \Delta \in X, \Delta' \in Y \} \\
    \subseteq{} & \cl{\{\Delta\csemic \Delta' \mid \Delta \in X, \Delta' \in Y \}}.
  \end{align*}
  From \Cref{lem:closed_set_laws} we get
  \[
    \Delta \in \cl{\{\Delta\csemic \Delta' \mid \Delta \in X, \Delta' \in Y \}}.
  \]

  For the inclusion from right to left: it suffices to show:
  \[
    \{\Delta\csemic \Delta' \mid \Delta \in X, \Delta' \in Y \} \subseteq X \cap Y.
  \]
  If $\Delta \in X$ and $\Delta' \in Y$, then $\Delta\csemic \Delta' \in X \cap Y$ by \Cref{lem:closed_set_laws}.
\end{proof}

\subsection{BI structure.}
In order to apply \Cref{thm:bi_alg_lifting_moore} and obtain a BI algebra structure on $\mathcal{C}$, we have to ensure that
the Heyting implication of closed sets is closed, and that $X \wandBullet Y \in \mathcal{C}$ whenever $Y \in \mathcal{C}$.

For the following lemma we will use the fact that the \ruleref{wandR} is invertible and \Cref{{lem:comprehension_adm}}.
\begin{lemma}[\coqident{cutelim}{wand_is_closed}]
  \label{lem:wand_eq}
  \begin{sloppypar}
    If $Y$ is closed, then so is ${X \wandBullet Y}$; furthermore, it can be described as:
  \end{sloppypar}
  \[
    X \wandBullet Y = \{ \Delta \mid \All \Delta' \in X. (\Delta\ccomma\Delta')\in Y \}.
  \]
\end{lemma} 
\begin{proof}
  It is straightforward to check that $X \wandBullet Y$ defined as above is indeed a right adjoint to the $\opBullet$ operation.
  Thus, it remains to show that $X \wandBullet Y$ is closed.

  Since $Y$ is closed, it can be written as an intersection of some family of principal closed sets:
   $Y = \bigcap_{j \in J} \outI{\varphi_j}$.
  Then, we claim,
  \[X \wandBullet Y = \bigcap_{(\Delta', j) \in X \times J} \outI{\frmlI{\Delta'} \wand \varphi_j}.\]
  
  \begin{sloppypar}
    \emph{Direction from left to right:} let $\Delta \in X \wandBullet Y$, and let ${(\Delta', j) \in X \times J}$.
    We are to show: ${\Delta \provesCF \frmlI{\Delta'} \wand \varphi_j}$.
    We argue as follows:
    \begin{mathpar}
      \infer*
      {\infer*
        {\Delta\ccomma \Delta' \provesCF \varphi_j}
        {\Delta\ccomma \frmlI{\Delta'} \provesCF \varphi_j}}
      {\Delta \provesCF \frmlI{\Delta'} \wand \varphi_j}
    \end{mathpar}
  \end{sloppypar}

  \emph{Direction from right to left:} suppose that
  $$\Delta \in \bigcap_{(\Delta', j) \in X \times J} \outI{\frmlI{\Delta'} \wand \varphi_j},$$
  and let $\Delta' \in X$.
  We are to show $\Delta\ccomma \Delta' \provesCF \varphi_j$ for any $j \in J$.
  By the assumption we have
  \[
    \Delta \provesCF \frmlI{\Delta'} \wand \varphi_j.
  \]
  We then reason similarly as in the previous direction, but using inversions \Cref{lem:wand_inv_adm,lem:comprehension_adm}:
  \begin{mathpar}
  \infer*
  {\infer*
    {\Delta \provesCF \frmlI{\Delta'} \wand \varphi_j}
    {\Delta\ccomma \frmlI{\Delta'} \provesCF \varphi_j}}
  {\Delta\ccomma \Delta' \provesCF \varphi_j}
\end{mathpar}
\end{proof}

We can give a similar characterization of the Heyting implication in $\mathcal{C}$:
\begin{proposition}[\coqident{cutelim}{has_heyting_impl}]
  For every closed $X, Y$, the Heyting implication is closed and
  can be described as:
  \[X \to Y = \{ \Delta \mid \All \Delta' \in X,\, (\Delta\csemic \Delta') \in Y\}.\]
\end{proposition}
\begin{proof}
  Using \Cref{{prop:meet_in_C}}, it is straightforward to check that $X \to Y$ as defined above is a right adjoint to the meet operation $\cap$.
  The proof of closedness follows the proof of \Cref{lem:wand_eq}.
\end{proof}

To sum up, by \Cref{thm:bi_alg_lifting_moore} we have a BI algebra $\mathcal{C}$ in which operations are defined as follows:

\begin{mathparpagebreakable}
  \EMP = \cl{\{\empM\}} \and
  \top = \Bunch \and\\
  \bot = \cl{\emptyset} \and
  X \vee Y = \cl{X \cup Y}
  \and
  X \ast Y = \cl{{\{\Delta\ccomma \Delta' \mid \Delta \in X, \Delta' \in Y \}}}
  \and
  X \wedge Y = \cl{\{\Delta\csemic \Delta' \mid \Delta \in X, \Delta' \in Y \}}
  \and
  X \wand Y = \{ \Delta \mid \All \Delta' \in X. (\Delta\ccomma\Delta')\in Y \}
  \and
  X \to Y  = \{ \Delta \mid \All \Delta' \in X. (\Delta\csemic\Delta')\in Y \}  
\end{mathparpagebreakable}

\subsection{Fundamental property of $\mathcal{C}$}
We can interpret formulas in the model ${\cal C}$ by picking the interpretation of atomic propositions to be $\Sem{a} = \outI{a}$.
Now we are ready to prove the main theorem: if $\Sem{\varphi} \subseteq \Sem{\psi}$, then $\varphi \provesCF \psi$.
To obtain this, we prove the following property, due to Okada~\cite{okada:99}.
\begin{lemma}[\coqident{cutelim}{okada_property}]
  \label{lem:fundamental}
  For any formula $\varphi$,
  \[
    \varphi \in \Sem{\varphi} \subseteq \outI{\varphi}
  \]
  (where the leftmost instance of $\varphi$ is a bunch consisting of a single leaf with the formula $\varphi$).
\end{lemma}
\begin{proof}
  By induction on $\varphi$.

  \emph{Case } $\varphi = \FALSE$.
  We have $\Sem{\FALSE} = \cl{\emptyset}$.
  Clearly, $\cl{\emptyset} \subseteq \outI{\varphi}$, because $\outI{\varphi}$ is closed and $\emptyset \subseteq \outI{\varphi}$.
  By \Cref{{lem:closed_set_laws}} we have $\FALSE \in \Sem{\FALSE}$.
  \smallskip
  
  \emph{Case } $\varphi = \TRUE$.
  We have $\Sem{\TRUE} = \Bunch = \outI{\TRUE}$.
  \smallskip

  \emph{Case } $\varphi = \EMP$.
  In order to show $\Sem{\EMP} = \cl{\{\empM\}} \subseteq \outI{\EMP}$, it suffices to show $\{\empM\} \subseteq \outI{\EMP}$, by the characterization of the closure operator.
  That inclusion holds because $\empM \provesCF \EMP$.
  In order to show $\EMP \in \cl{\{\empM\}}$, it suffices to show $\empM \in \cl{\{\empM\}}$ by \Cref{lem:closed_set_laws}, which holds trivially.
  \smallskip

  \emph{Case } $\varphi = \psi_1 \ast \psi_2$.
  In order to show the set inclusion $\Sem{\psi_1 \ast \psi_2} = \cl{\Sem{\psi_1} \opBullet \Sem{\psi_2}} \subseteq \outI{\psi_1 \ast \psi_2}$,
  it suffices to show $\Sem{\psi_1} \opBullet \Sem{\psi_2} \subseteq \outI{\psi_1 \ast \psi_2}$, by the characterization of the closure operator.
  If $(\Delta_1, \Delta_2) \in \Sem{\psi_1} \opBullet \Sem{\psi_2}$, then, by the induction hypothesis $\Delta_i \provesCF \psi_i$, and we can reason as follows:
  \begin{mathpar}
  {\infer
    {\Delta_1 \provesCF \psi_1 \and \Delta_2 \provesCF \psi_2}
    {\Delta_1, \Delta_2  \provesCF \psi_1 \ast \psi_2}}
  \end{mathpar}
  Hence, $(\Delta_1, \Delta_2) \in \outI{\psi_1 \ast \psi_2}$.

  As for the element inclusion $\psi_1 \ast \psi_2 \in \cl{\Sem{\psi_1} \opBullet \Sem{\psi_2}}$, note that by \Cref{lem:closed_set_laws} it suffices to show $(\psi_1, \psi_2) \in \cl{\Sem{\psi_1} \opBullet \Sem{\psi_2}}$, which is evident from the induction hypotheses.
  \smallskip

  \emph{Case } $\varphi = \psi_1 \wedge \psi_2$.
  In order to show the set inclusion, suppose that $\Delta \in \Sem{\psi_1 \wedge \psi_2} = \Sem{\psi_1} \cap \Sem{\psi_2}$.
  Then, by the induction hypothesis, $\Delta \in \outI{\psi_1} \cap \outI{\psi_2}$, and we can reason as follows:
  \begin{mathpar}
  \infer*
  {\infer
    {\Delta \provesCF \psi_1 \and \Delta \provesCF \psi_2}
    {\Delta\csemic \Delta \provesCF \psi_1 \wedge \psi_2}}
  {\Delta \provesCF \psi_1 \wedge \psi_2}
  \end{mathpar}

  As for the element inclusion $\psi_1 \wedge \psi_2 \in \Sem{\psi_1} \cap \Sem{\psi_2}$, we argue as follows.
  By the induction hypothesis, $\psi_1 \in \Sem{\psi_1}$.\
  By \Cref{lem:closed_set_laws} (item 1), $(\psi_1 ; \psi_2) \in \Sem{\psi_1}$, and by \Cref{lem:closed_set_laws} (item 3),
  $\psi_1 \wedge \psi_2 \in \Sem{\psi_1}$.
  Similarly we can show $\psi_1 \wedge \psi_2 \in \Sem{\psi_2}$.
  \smallskip

  \emph{Case } $\varphi = \psi_1 \wand \psi_2$.
  In order to show $\Sem{\psi_1 \wand \psi_2} = \Sem{\psi_1} \wand \Sem{\psi_2} \subseteq \outI{\psi_1 \wand \psi_2}$, suppose that
  $\Delta \in \Sem{\psi_1} \wand \Sem{\psi_2}$.
  We are to show $\Delta \provesCF \psi_1 \wand \psi_2$.
  By the induction hypothesis, $\psi_1 \in \Sem{\psi_1}$; hence
  \[
    (\Delta\ccomma \psi_1) \in \Sem{\psi_2} \subseteq \outI{\psi_2}.
  \]
  We can then reason using the right rule for $\wand$:
  \begin{mathpar}
  \infer
  {\Delta\ccomma \psi_1 \provesCF \psi_2}
  {\Delta \provesCF \psi_1 \wand \psi_2}
  \end{mathpar}

  In order to show $\psi_1 \wand \psi_2 \in \Sem{\psi_1} \wand \Sem{\psi_2}$, suppose that $\Delta \in \Sem{\psi_1}$.
  We are to show $(\Delta\ccomma \psi_1 \wand \psi_2) \in \Sem{\psi_2}$.
  Let us write $\Sem{\psi_2}$ as $\bigcap_{i \in I}\outI{\varphi_i}$.
  Then our goal can be reduced to showing
  \[
    \Delta\ccomma \psi_1 \wand \psi_2 \provesCF \varphi_i
  \]
  for any $i \in I$.
  We argue as follows, using the left rule for $\wand$:
  \begin{mathpar}
  \infer
  {\Delta \provesCF \psi_1 \and
    \psi_2 \provesCF \varphi_i
  }
  {\Delta\ccomma \psi_1 \wand \psi_2 \provesCF \varphi_i}
  \end{mathpar}
  where the first assumption holds because $\Delta \in \Sem{\psi_1} \subseteq \outI{\psi_1}$ and the second assumption holds because $\psi_2 \in \outI{\psi_2}$.

  \smallskip
  \emph{Case } $\varphi = \psi_1 \to \psi_2$. Similarly to the case $\varphi = \psi_1 \wand \psi_2$, using the characterization of the Heyting implication in $\mathcal{C}$.
  \smallskip

  \emph{Case } $\varphi = \psi_1 \vee \psi_2$.
  In order to show $\Sem{\psi_1 \vee \psi_2} = \Sem{\psi_1}\vee \Sem{\psi_2} \subseteq \outI{\psi_1 \vee \psi_2}$, it suffices to show $\Sem{\psi_1} \subseteq \outI{\psi_1 \vee \psi_2}$ and $\Sem{\psi_2} \subseteq \outI{\psi_1 \vee \psi_2}$.
  To show that $\Sem{\psi_i} \subseteq \outI{\psi_1 \vee \psi_2}$, for $i = 1, 2$, it suffices to show $\outI{\psi_i} \subseteq \outI{\psi_1 \vee \psi_2}$.
  We show that using the right rules for disjunction.

  To show $\psi_1 \vee \psi_2 \in \Sem{\psi_1 \vee \psi_2} = \cl{\Sem{\psi_1} \cup \Sem{\psi_2}}$, we appeal to the definition of $\cl{-}$:
  Let $\varphi$ be a formula such that $\Sem{\psi_1} \cup \Sem{\psi_2} \subseteq \outI{\varphi}$.
  We are to show $\psi_1 \vee \psi_2 \in \outI{\varphi}$, i.e.~$\psi_1 \vee \psi_2 \provesCF \varphi$.
  By assumption we have $\psi_i \in \Sem{\psi_i}$, for $i = 1,2$, and, hence
  $\psi_i \in \outI{\varphi}$, or, equivalently, $\psi_i \provesCF \varphi$.
  We obtain the desired result using \ruleref{disjL}.
\end{proof}

\begin{theorem}[\coqident{cutelim}{C_interp_cf}]
  \label{thm:c_interp_cf}
  If $\Sem{\frmlI{\Delta}} \leq \Sem{\varphi}$ holds in $\mathcal{C}$, then $\Delta \provesCF \varphi$.
\end{theorem}
\begin{proof}
  By \Cref{lem:fundamental}, we have $\frmlI{\Delta} \in \Sem{\frmlI{\Delta}}$, and hence ${\frmlI{\Delta} \in \Sem{\varphi}}$.
  By \Cref{lem:closed_set_laws} we have furthermore have $\Delta \in \Sem{\varphi}$ which is equivalent to  $\Delta \provesCF \varphi$.
\end{proof}

As a consequence, we get the cut admissibility:
\begin{theorem}[\coqident{cutelim}{cut}]
  The \ruleref{cut} rule is admissible in the cut-free fragment $\provesCF$ of BI.
\end{theorem}
\begin{proof}
  Suppose $\Delta \provesCF \psi$ and $\Gamma(\psi) \provesCF \varphi$.
  We are to show that $\Gamma(\Delta) \provesCF \varphi$.
  By \Cref{thm:c_interp_cf} it suffices to show that 
  $\Sem{\frmlI{\Gamma(\Delta)}} \leq \Sem{\varphi}$ holds in $\mathcal{C}$.

  From soundness we have that $\Sem{\frmlI{\Delta}} \leq \Sem{\psi}$.
  By induction on $\Gamma$ we can show that $\Sem{\frmlI{\Gamma(\Delta)}} \leq \Sem{\frmlI{\Gamma(\psi)}}$, from which we obtain
  \[
    \Sem{\frmlI{\Gamma(\Delta)}} \leq \Sem{\frmlI{\Gamma(\psi)}}
    \leq \Sem{\varphi}.
  \]
\end{proof}


\paragraph{Overview.}
In the next sections we will be looking at adjusting the construction of $\mathcal{C}$ to extensions of BI.
At this point we would like to give an overview of the argument, and see what kind of conditions we need.

\begin{itemize}
\item To show that the closure operator $\cl{-}$ is strong, we had to use invertibility of certain rules.
  Firstly, we used the fact that BI satisfies a strong form of the deduction theorem for both implications (the rules \ruleref{implR} and \ruleref{wandR} are invertible).
  Secondly, we used the fact that the left rules are invertible for connectives that form bunches (\ruleref{empL}, \ruleref{trueL}, \ruleref{andL}, \ruleref{sepL}).
\item Additionally, we need to verify that all the rules of sequent calculus are validated in $\mathcal{C}$.
\item
  \begin{sloppypar}
    Finally, we need to show that Okada's property (\Cref{lem:fundamental}) holds in $\mathcal{C}$.
  \end{sloppypar}
\end{itemize}

This list gives us a sort of roadmap for extending the cut elimination argument.
For every rule that we want to add to BI, we need to re-verify the invertibility of certain rules, and that the rule is validated in $\mathcal{C}$.
If we want to add a new connective to the system, we need to additionally come up with the interpretation of this connective on $\mathcal{C}$, and re-verify Okada's property.


\section{Extending the logic: simple structural rules}
\label{sec:simple_ext}
An important extension of BI is \emph{affine BI}, which extends the sequent calculus of \Cref{fig:bi_seqcalc} with the weakening rule for $\ccomma$:
\begin{mathpar}
\inferhref{W$\ccomma$}{W,}
{\Delta(\Delta_1) \proves \varphi}
{\Delta(\Delta_1 \ccomma \Delta_2) \proves \varphi}
\end{mathpar}
An algebraic structure for interpreting affine BI is a BI algebra in which the following inequality holds: $p \ast q \leq p$.
Can we extend the argument presented so far to cover affine BI?
As we discussed at the end of the previous section, because we are adding a new rule, we have to make sure that the analogues of \Cref{lem:wand_inv_adm} and \Cref{lem:inv_rules} still hold (the appropriate rules are invertible), and that $\mathcal{C}$ validates the inequality $X \ast Y \subseteq X$.

To verify that $X \ast Y \subseteq X$ it suffices to verify that $X \opBullet Y \subseteq X$, since $X$ is closed.
Let us write $X = \bigcap_{i \in I} \outI{\varphi_i}$.
Suppose that $\Delta_1 \in X, \Delta_2 \in Y$.
We are to show that $\Delta_1 \ccomma \Delta_2 \provesCF \varphi_i$ for any $i$;
however we know that $\Delta_1 \provesCF \varphi_i$ by the assumption, and the desired result follows by \ruleref{W,}.

This kind of argument for \ruleref{W,} can be generalized to infinitely many structural rules of a particular shape, which we call, following \cite{galatos.jipsen:2013}, \emph{simple structural rules}.
In the remainder of this section we show how to define such simple structural rules, and we prove cut elimination for BI extended with any combination of such rules.

\subsection{Simple structural rules and bunched terms}
Simple structural rules are rules of the shape
\begin{mathpar}
\infer
{\Pi(T_1[\Delta_1, \dots, \Delta_n]) \proves \varphi \and
\dots \and
\Pi(T_m[\Delta_1, \dots, \Delta_n]) \proves \varphi
}
{\Pi(T[\Delta_1, \dots, \Delta_n]) \proves \varphi}
\end{mathpar}
where $T_1, \dots, T_m, T$ are \emph{bunched terms} -- bunches built out of connectives $\ccomma, \csemic$, and variables $x_1, \dots, x_n$.
The notation $T_i[\vec{\Delta}]$ stands for the bunch obtained from $T_i$ by replacing all the variables $x_j$ with $\Delta_j$.
Furthermore, in the rule above we require that $T$ is a \emph{linear} bunched term -- a term in which every variable $x_j$ occurs at most once.

We identify a structural rule with a tuple $(\{T_1, \dots, T_m\}, T)$.
The rule \ruleref{W,} above is represented with a tuple $(\{ x_1 \}, x_1\ccomma x_2)$.
If $L$ is a set of such tuples, we write BI+$L$ for a sequent calculus of BI extended with the structural rules from $L$.

For the rest of this section, we fix a finite collection $L$ of simple structural rules and the extended system BI+$L$.
We write $\provesCF$ for cut-free provability in BI+$L$, and we denote by $\mathcal{C}$ the BI algebra
constructed in \Cref{{sec:cutelim_thm}}, but for BI+$L$-provability.

Firstly, we need to check that the construction of $\mathcal{C}$ works out.
We need to verify that the rules \ruleref{implL}, \ruleref{wandL}, \ruleref{andL}, \ruleref{sepL}, \ruleref{trueL}, \ruleref{empL} are still invertible, in presence of the additional rules from $L$.
For that, we just follow the proof of \Cref{lem:inv_rules}.

\subsection{Interpretation of simple structural rules in $\mathcal{C}$}
Additionally, we need to verify that $\mathcal{C}$ validates all the rules from $L$.

Each bunched term $T[x_1, \dots, x_n]$ can be interpreted as a function $\Sem{T} : A^n \to A$ on any BI algebra $A$.
For example, a (non-linear) bunched term $(x_1 \ccomma x_2) \csemic x_1$ gives rise to a mapping $(X_1, X_2) \mapsto (X_1 \ast X_2) \wedge X_1$.

In order to interpret a simple structural rule given by a tuple ${(\{ T_1, \dots, T_m \}, T)}$ in a BI algebra $A$, we require that the following inequality holds in $A$ for any $a_1, \dots, a_n \in A$:
\[
  \Sem{T}(a_1, \dots, a_n) \leq \Sem{T_1}(a_1, \dots, a_n) \vee \dots \vee \Sem{T_m}(a_1, \dots, a_n).
\]
In this case, we say that $A$ validates the simple structural rule.
For example, recall that the weakening rule \ruleref{W,} for $\ccomma$ is represented as $(\{x_1\}, (x_1\ccomma x_2))$.
Then the associated inequality is:
\[
  \Sem{x_1\ccomma x_2}(p, q) \leq \Sem{x_1}(p, q) \quad \iff \quad
p \ast q \leq p.
\]
\begin{lemma}[\coqmod{seqcalc}{Seqcalc}{seq_interp_sound}]
\label{lem:ext_soundness}
  If a BI algebra $A$ validates the rules in $L$, then $\Delta \proves \varphi$ implies $\Sem{\Delta} \leq \Sem{\varphi}$ in $A$.
\end{lemma}
\begin{proof}
  For the case of a simple structural rule $(\{ T_1, \dots, T_m \}, T)$, we assume that
  $\Sem{T_i}(\bar{a}) \leq \Sem{\varphi}$ holds for any $1 \leq i \leq m$.
  Then, $\bigvee_{1 \leq i \leq m} \Sem{T_i}(\bar{a}) \leq \Sem{\varphi}$.
  Since the rule is validated in $A$ we have
  \[
    \Sem{T}(\bar{a}) \leq \bigvee_{1 \leq i \leq m} \Sem{T_i}(\bar{a}) \leq \Sem{\varphi}.
  \]
\end{proof}

In order to show that $\mathcal{C}$ validates all the rules from $L$, we need the following lemmas about $\Sem{T}$.
For the algebra $\mathcal{C}$ we have the following description:
\begin{lemma}[\coqident{cutelim}{bterm_C_refl}]
  \label{lem:bterm_c}
  Let $X_1, \dots, X_n \in \mathcal{C}$, and $\Delta_i \in X_i$ for $1\leq i \leq n$.
  Then for any bunched term $T$,
  \[
    T[\vec{\Delta}] \in \Sem{T}(\vec{X})
  \]
\end{lemma}
\begin{proof}
  By induction on $T$.
\end{proof}
\begin{lemma}[\coqident{cutelim}{blinterm_C_desc'}]
  \label{lem:lin_bterm}
  For any $X_1, \dots, X_n \in \mathcal{C}$ and any linear bunched term $T$ we have
  \[
    \Sem{T}(X_1, \dots, X_n) = \cl{\{ T[\Delta_1, \dots, \Delta_n] \mid \Delta_i \in X_i, 1 \leq i \leq n\}}
  \]
\end{lemma}
\begin{proof}
  In view of \Cref{lem:bterm_c} it suffices to show that the left-hand side is included in the right-hand side.
  This is done by induction on $T$.
  We show only the case for $\ccomma$, as the other case is similar.
  If $T(\vec{x}) = F(\vec{x})\ccomma U(\vec{x})$, then, since $T$ is linear, we can write it down as
  \[
    T(\vec{y}\vec{z}) = F(\vec{y})\ccomma U(\vec{z})
  \]
  for some factorization $\vec{y}\vec{z} = \vec{x}$, and for some linear terms $F$ and $U$.
  By the induction hypothesis we have
  \[
    \Sem{T}(\vec{Y}\vec{Z}) = \cl{\cl{\{F[\vec{\Gamma}] \mid \vec{\Gamma}\in \vec{Y}\}} \opBullet \cl{\{U[\vec{\Gamma}] \mid \vec{\Gamma}\in \vec{Z}\}}}.
  \]
  In order to show the inclusion into $\cl{\{T[\vec{\Delta}] \mid \vec{\Delta} \in \vec{Y}\vec{Z}\}}$ it suffices to show
  \[
    \{F[\vec{\Gamma}] \mid \vec{\Gamma}\in \vec{Y}\} \opBullet \{U[\vec{\Gamma}] \mid \vec{\Gamma}\in \vec{Z}\}
    \subseteq \{T[\vec{\Delta}] \mid \vec{\Delta} \in \vec{Y}\vec{Z}\}.
  \]
  Let $\vec{\Gamma}\in\vec{Y}$ and $\vec{\Theta} \in \vec{Z}$.
  Then, $\vec{\Gamma}\vec{\Theta} \in \vec{X}$, and, hence
  $F[\vec{\Gamma}] \ccomma U[\vec{\Theta}] = T[\vec{\Gamma}\vec{\Theta}]$,
  which concludes the proof the inclusion.
\end{proof}

With the two lemmas at hand we can prove that $\mathcal{C}$ is a model of BI+$L$.
\begin{lemma}[\coqident{cutelim}{C_extensions}]
\label{lem:C_model_ext}
Every rule from the set $L$ is validated in $\mathcal{C}$.
\end{lemma}
\begin{proof}
  Suppose that $(\{T_1, \dots, T_m\}, T)$ is a simple structural rule from $L$.
  We have to show $\Sem{T}(\vec{X}) \subseteq \cl{\bigcup_{1 \leq i \leq m}\Sem{T_i}(\vec{X})}$.
  By \Cref{lem:lin_bterm}, it suffices to show
  \[
    \{ T[\Delta_1, \dots, \Delta_n] \mid \vec{\Delta}\in\vec{X}\} \subseteq \cl{\bigcup_{1 \leq i \leq m}\Sem{T_i}(\vec{X})}
  \]
  where $\vec{\Delta}\in\vec{X}$ is a shorthand for $\Delta_i \in X_i$ for all $1 \leq i \leq n$.
  
  Suppose that $\varphi$ is such that $\bigcup_{1 \leq i \leq n}\Sem{T_i}(\vec{X}) \subseteq \outI{\varphi}$.
  We are to show that $T[\vec{\Delta}] \provesCF \varphi$, for any $\vec{\Delta} \in \vec{X}$.
  By \Cref{lem:bterm_c}, we have $T_i[\vec{\Delta}] \in \Sem{T_i}(\vec{X})\subseteq \outI{\varphi}$.
  So we get $T_i[\vec{\Delta}] \provesCF \varphi$, from which we can conclude that
  $T[\vec{\Delta}] \provesCF \varphi$
\end{proof}

\begin{theorem}[\coqident{cutelim}{cut}]
  The \ruleref{cut} rule is admissible in the cut-free fragment $\provesCF$ of BI+$L$.
\end{theorem}

\section{Extending the logic: an S4 modality}
\label{sec:modal_ext}
In this section we look at a different kind of extension to BI, the one obtained by ``freely'' adding an (intuitionistic) S4-like modality.
This amounts to adding the following rules (usual for intuitionistic formulation of S4 sequent calculus \cite{bierman.depaiva:2000}):
\begin{mathpar}
\inferhref{$\Box$R}{boxR}
  {\Box \Delta \proves A}
  {\Box \Delta \proves \Box A}
\and
\inferhref{$\Box$L}{boxL}
  {\Delta(A) \proves B}
  {\Delta(\Box A) \proves B}
\end{mathpar}
where $\Box \Delta$ is the same as $\Delta$, but with boxes $\Box$ put in front of all the formulas, e.g.
\[
  \Box (\empM\csemic (\varphi\ccomma \psi)\csemic \chi) \eqdef
  \empM\csemic (\Box \varphi\ccomma \Box \psi)\csemic \Box \chi.
\]
We denote the extended system (the sequent calculus from \Cref{fig:bi_seqcalc} with the rules \ruleref{boxR}, \ruleref{boxL} above) as BIS4.
We can verify that the relevant rules are still invertible (a version of \Cref{{lem:inv_rules}} and \Cref{lem:wand_inv_adm} for BIS4).

\paragraph{Interpreting the modality.}
As per the roadmap at the end of \Cref{{sec:cutelim_thm}} we need to interpret the modality $\Box$ on $\mathcal{C}$ somehow.
The usual way of interpreting a $\Box$ modality in intuitionistic setting is with an interior operator (c.f. the notion of a CS4 algebra \cite[Definition 3]{alechina.etal:2001}).
\begin{definition}
\label{def:BIS4_alg}
A \emph{BIS4 algebra} is a tuple $({\cal B}, \Box)$ where ${\cal B}$ is a BI algebra and $\Box : {\cal B} \to {\cal B}$ is a monotone function satisfying:
\begin{enumerate}
\item\label{eq:box_elim} $\Box p \leq p$;
\item\label{eq:box_idem} $\Box p \leq \Box \Box p$;
\item\label{eq:box_top} $\top = \Box \top$
\item\label{eq:box_emp} $\EMP = \Box \EMP$;
\item\label{eq:box_conj} $\Box p \wedge \Box q \leq \Box (p \wedge q)$;
\item\label{eq:box_sep} $\Box p \ast \Box q \leq \Box (p \ast q)$.
\end{enumerate}
\end{definition}

We define the interior operator $\Box$ on ${\cal C}$ as:
\[
  \Box X \eqdef \cl{\{ \Box \Delta \mid \Delta \in X \}}.
\]

In order to show that ${\cal C}$ satisfies the conditions from \Cref{def:BIS4_alg}, we will use the following lemmas.
\begin{lemma}[\coqident{seqcalc_height_s4}{box_l_inv}]
  \label{lem:box_inversion_l}
  The following rule is admissible:
  \begin{mathpar}
  \inferhref{$\Box$-idemp}{box-idemp}
  {\Gamma(\Box \Box \Delta) \proves \varphi}
  {\Gamma(\Box \Delta) \proves \varphi}
  \end{mathpar}
\end{lemma}
\begin{proof}
By induction on the height of the derivation, similar to the proof of \Cref{lem:inv_rules}.
\end{proof}
\begin{lemma}[\coqmod{cutelim_s4}{Cl}{C_necessitate}, \coqmod{cutelim_s4}{Cl}{C_bunch_box_idemp}]
\label{lem:closed_set_laws_box}
Let $X$ be a closed set.
\begin{itemize}
\item If $\Delta \in X$, then $\Box \Delta \in X$.
\item If $\Box \Box \Delta \in X$, then $\Box \Delta \in X$.
\end{itemize}
\end{lemma}
\begin{proof}
  By examining the definitions of $\Box$ and $\cl{-}$, using \Cref{lem:box_inversion_l} for the second item.
\end{proof}

\begin{lemma}[\coqmod{cutelim_s4}{Cl}{C_alg_box}]
  $({\cal C}, \Box)$ is a BIS4 algebra.
\end{lemma}
\begin{proof}
The conditions (\ref{eq:box_elim}) and (\ref{eq:box_idem}) follow from \Cref{{lem:closed_set_laws_box}}.
The conditions (\ref{eq:box_top}) and (\ref{eq:box_emp}) can be shown by examining the definitions of all the connectives involved.

The condition (\ref{eq:box_sep}) can be shown as follows.
To show $\Box X \ast \Box Y \subseteq \Box (X \ast Y)$,
we reason as follows:
\begin{align*}
  &\Box X \ast \Box Y =
  \cl{\cl{\{\Box \Delta \mid \Delta \in X\}} \opBullet \cl{\{\Box \Delta \mid \Delta \in Y\}}} \subseteq{}\\
  &\cl{\cl{\{\Box \Delta \mid \Delta \in X\} \opBullet \cl{\{\Box \Delta \mid \Delta \in Y\}}}} =\\
  &\cl{\{\Box \Delta \mid \Delta \in X\} \opBullet \cl{\{\Box \Delta \mid \Delta \in Y\}}}.
\end{align*}
To show that 
\[
\cl{\{\Box \Delta \mid \Delta \in X\} \opBullet \cl{\{\Box \Delta \mid \Delta \in Y\}}} \subseteq \Box(X \ast Y)
\]
it suffices to show that
\[
\{\Box \Delta \mid \Delta \in X\} \opBullet \cl{\{\Box \Delta \mid \Delta \in Y\}} \subseteq \Box(X \ast Y).
\]
And, since
\begin{multline*}
  \{\Box \Delta \mid \Delta \in X\} \opBullet \cl{\{\Box \Delta \mid \Delta \in Y\}} \\
  \subseteq \cl{\{\Box \Delta \mid \Delta \in X\} \opBullet \{\Box \Delta \mid \Delta \in Y\}},
\end{multline*}
it suffices to show
\[
{\{\Box \Delta \mid \Delta \in X\} \opBullet \{\Box \Delta \mid \Delta \in Y\}} 
\subseteq \Box(X \ast Y).
\]
Let $\Delta$ be such that $\Delta = \Box \Delta_1 \ccomma \Box \Delta_2$, for $\Delta_1 \in X$, $\Delta_2 \in Y$.
Then $\Delta = \Box (\Delta_1 \ccomma \Delta_2)$, with $\Delta_1 \ccomma \Delta_2 \in X \ast Y$.

Finally, the condition (\ref{eq:box_conj}) is shown similarly.
\end{proof}

All it remains to verify is that Okada's property (\Cref{lem:fundamental}) still holds.
Since we have added only the $\Box$ modality we need to check one extra case:
\begin{lemma}
  Assume that $\varphi$ is such that $\varphi \in \Sem{\varphi} \subseteq \outI{\varphi}$.
  Then
  \[
    \Box \varphi \in \Sem{\Box \varphi} \subseteq \outI{\Box\varphi}.
  \]
\end{lemma}
\begin{proof}
In order to show the first inclusion, note that
by the hypothesis, we have $\varphi \in \Sem{\varphi}$.
Hence, $$\Box \varphi \in \{ \Box \Delta \mid \Delta \in \Sem{\varphi} \} \subseteq \Box \Sem{\varphi}.$$

To show the second inclusion it suffices to show
\[
  \{\Box \Delta \mid \Delta \in \Sem{\varphi}\} \subseteq \outI{\Box\varphi}.
\]
So, let us assume $\Delta \in \Sem{\varphi}$.
By the induction hypothesis we have $\Delta \provesCF \varphi$, and, hence $\Box \Delta \provesCF \Box \varphi$.
Which gives us the desired result $\Box \Delta \in \outI{\Box \varphi}$.
\end{proof}

\begin{theorem}[\href{\coqdocurl{cutelim_s4}{cut}}{\nolinkcoqident{cutelim_s4.cut}}]
  The \ruleref{cut} rule is admissible in the cut-free fragment $\provesCF$ of BIS4.
\end{theorem}


\section{The Coq formalization}
\label{sec:formalization}
As we mentioned, the results of this paper has been formalized in the Coq proof assistant.
In this section we describe some of the design choices and trade-offs that we made.

Instead of formalizing sequent calculus with the cut rule and deriving a cut-free sequence calculus from that, we opted for formalizing just the cut-free sequent calculus and proving that cut it admissible in that system.
The sequent calculus (and, consequently, the algebra $\mathcal{C}$) is parameterized by a collection of simple structural rules (as in \Cref{sec:simple_ext}), which is represented in Coq as a module of the 
 \href{\coqdocurl{seqcalc}{SIMPLE_STRUCT_EXT}}{following signature}:
\begin{lstlisting}[language=Coq]
Module Type SIMPLE_STRUCT_EXT.
  Definition bterm := bterm nat.
  Parameter rules :
    list (list bterm * bterm).
  Parameter rules_good :
    ∀ (Ts : list bterm) (T : bterm),
        (Ts, T) ∈ rules → linear_bterm T.
End SIMPLE_STRUCT_EXT.
\end{lstlisting}
The type \coqe{bterm} represents bunched terms,
and each simple structural rule is given as a tuple \coqe{(Ts, T)} of bunched terms in the premises and in the conclusion.

As for the algebraic semantics, we used a slightly modified formalization of BI algebras from the Iris Coq library~\cite{irisWWW,MoSeL}.
The original formulation BI algebras in Iris also includes a \emph{persistence modality}~\cite{bizjak.birkedal:2018}, which behaves quite differently from an S4-like modality that we use in \Cref{sec:modal_ext}.
To our knowledge, the proof theory of this modality has not been studied and there is no sequent calculus for this logic.
The Iris formalization makes heavy use of setoids, which allows us to easily formulate the model $\pset{\Bunch}$ of predicates on bunches quotiented by bunch equivalence.

The trickiest proofs to formalize were the admissibility of the inverted rules (\Cref{{lem:inv_rules}}) in the cut-free sequent calculus.
Firstly, as was mentioned in \Cref{{sec:seqcalc}}, those admissibility proofs proceed by induction on the height of the derivation.
To handle this in the Coq formalization, we use an auxiliary relation
\href{\coqdocurl{seqcalc_height}{SeqcalcHeight.proves}}{\coqe{proves : bunch → formula → nat → Prop}}
which includes the (upper bound on the) height of the derivation.
Our reasoning behind this definition is that if we were to define a proof height function and do induction on its value, we would have to formulate our goal (and the proof) in a rather unwieldy way: we would have to package together the context, the formula, and the derivation into a $\Sigma$-type:
\coqe{Σ ($\Delta$ : bunch) ($\varphi$ : formula), proves $\Delta$ $\varphi$}.

Secondly, even with induction on proof height, in the proof of \Cref{{lem:inv_rules}} we often end with a situation where we have a bunch $\Delta$ that can be decomposed multiple ways that we need to related to each other.
For example, in the proof of invertibility of \ruleref{sepL}, we want to obtain a proof of $\Delta_0(\varphi\ccomma \psi) \proves \chi$ from a proof of $\Delta_0(\varphi \ast \psi) \proves \chi$.
Suppose that the last applied rule in the proof was weakening
\begin{mathpar}
\infer*
{\Delta_1(\Gamma_1) \provesCF \chi}
{\Delta_1(\Gamma_1\csemic \Gamma_2) \provesCF \chi}
\end{mathpar}
with $\Delta_1(\Gamma_1 \csemic \Gamma_2) = \Delta_0(\varphi \ast \psi)$.
In order to apply the induction hypothesis we have to locate the formula $\varphi \ast \psi$ somewhere in the bunch $\Delta_1(\Gamma_1)$.
The formula may appear either in $\Gamma_1$, $\Gamma_2$, or be part of the bunched context $\Delta_1(\cdot)$, depending on the relation between $\Delta_0$ and $\Delta_1$.
This is an example of informal observation that comes up often in the BI sequent calculus because all the left rules (and structural rules) can be applied deep inside an arbitrary bunch.
As such, reasoning about what appears where in bunched contexts is of importance.

In order to reason about situations like this in Coq, we define an \href{\coqdocurl{bunch_decomp}{bunch_decomp}}{auxiliary inductive system} $\bunchDecomp{\Delta}{\Pi(-)}{\Delta'}$ that captures exactly when $\Delta = \Pi(\Delta')$.
\begin{figure}[t]
  \begin{mathpar}
  \axiom
  {\bunchDecomp{\Delta}{(-)}{\Delta}} 
  \and
  \infer
  {\bunchDecomp{\Delta_1}{\Pi(-)}{\Delta}}
  {\bunchDecomp{\Delta_1\ccomma\Delta_2}{\Pi(-)\ccomma\Delta_2}{\Delta}}
  \and
  \infer
  {\bunchDecomp{\Delta_2}{\Pi(-)}{\Delta}}
  {\bunchDecomp{\Delta_1\ccomma\Delta_2}{\Delta_1\ccomma\Pi(-)}{\Delta}}
  \and
  \infer
  {\bunchDecomp{\Delta_1}{\Pi(-)}{\Delta}}
  {\bunchDecomp{\Delta_1\csemic\Delta_2}{\Pi(-)\csemic\Delta_2}{\Delta}}
  \and
  \infer
  {\bunchDecomp{\Delta_2}{\Pi(-)}{\Delta}}
  {\bunchDecomp{\Delta_1\csemic\Delta_2}{\Delta_1\csemic\Pi(-)}{\Delta}}
  \end{mathpar}
  \caption{Inductive rules for decomposition of bunches.}
  \label{fig:bunchDecomp}
\end{figure}
The rules for the decomposition of bunches is given in \Cref{{fig:bunchDecomp}}.
\begin{lemma}[\coqident{bunch_decomp}{bunch_decomp_iff}]
  $\Delta = \Pi(\Delta')$ if and only if $\bunchDecomp{\Delta}{\Pi}{\Delta'}$.
\end{lemma}
Using this inductive system we can prove the following lemmas about decomposition of contexts, that we use for formalizing proofs from \Cref{sec:seqcalc}:
\begin{lemma}[\coqident{bunch_decomp}{fill_is_frml}]
\label{lem:ctx_decomp_further}
If $\Pi(\Delta) = \varphi$ then $\Pi$ is an empty context and {$\Delta = \varphi$}.
\end{lemma}
\begin{lemma}[\coqident{bunch_decomp}{bunch_decomp_ctx}]
If $\bunchDecomp{\Pi(\Delta)}{\Pi'(-)}{\varphi}$
then one of the two conditions hold:
\begin{itemize}[leftmargin=*]
\item The formula $\varphi$ appears in $\Delta$ itself.
  That is, there is $\Pi_0(-)$ such that $\bunchDecomp{\Delta}{\Pi_0(-)}{\varphi}$ and $\Pi'(-) = \Pi(\Pi_0(-))$.
\item Or the formula $\varphi$ appears in the context $\Pi(-)$.
  Then we can think of $\Pi'(-)$ as a context with two holes, one of which is already filled with $\Delta$.
  Formally we represent this situation as follows.
  There are functions $\Pi_0, \Pi_1$ from bunches to bunched contexts, such that:
  \begin{itemize}
  \item For any bunch $\Lambda$, we have
    $\bunchDecomp{\Pi(\Lambda)}{\Pi_0(\Lambda)(-)}{\varphi}.$
  \item For any bunch $\Lambda$, we have
    $\bunchDecomp{\Pi'(\Lambda)}{\Pi_1(\Lambda)(-)}{\Delta}.$
  \item For any bunches $\Lambda, \Lambda'$, we have
    $\Pi_0(\Lambda)(\Lambda') = \Pi_1(\Lambda')(\Lambda).$
  \end{itemize}
\end{itemize}
\end{lemma}
Similarly, in order to prove the invertibility of relevant rules for the extension of BI with a set of simple structural rules (as in \Cref{{sec:simple_ext}}), we additionally make use of the following auxiliary lemma:
\begin{lemma}[\coqident{bunch_decomp}{bterm_ctx_act_decomp}]
  \begin{sloppypar}
    If $T$ is a linear bunched term with variables $x_1, \dots, x_n$, and $T[\vec{\Delta}] = \Pi(\varphi)$ for some bunched context $\Pi$, then
    there is a variable $x_j$ occurring in $T$, and a context $\Pi'$ such that
  \end{sloppypar}
  \begin{itemize}
\item $\Delta_j = \Pi'(\varphi)$;
\item for any bunch $\Gamma$,
  \[
    T[\Delta_1, \dots, \Delta_{j-1}, \Pi'(\Gamma), \Delta_{j+1}, \dots, \Delta_n] = \Pi(\Gamma).
  \]
\end{itemize}
\end{lemma}
In order to prove the invertibility of relevant rules for BIS4 (\Cref{sec:modal_ext}), including \Cref{{lem:box_inversion_l}} we make use of the following auxiliary lemma:
\begin{lemma}[\coqident{seqcalc_s4}{bunch_decomp_box}]
  If $\Box \Delta = \Pi(\Box \varphi)$, then there is
  a bunched context $\Pi'$ such that 
  \begin{itemize}
  \item $\Delta = \Pi'(\varphi)$;
  \item for any $\Gamma$, $\Box \Pi'(\Gamma) = \Pi(\Box \Gamma)$.
  \end{itemize}
\end{lemma}

\section{Related work}
\label{sec:related_work}
There has been a long line of work on formalizing cut elimination and other meta-theoretical properties of logics in proof assistants.
Here, we mention a few recent ones.
Pfenning~\cite{pfenning:2000} formalized cut elimination for intuitionistic and classical propositional logic in Elf, using only structural induction and avoiding termination measures.
Chaudhuri, Lima, and Reis~\cite{chaudhuri.etal:2017} have formalized cut elimination for various fragments of linear logic in Abella.
Xavier, Olarte, Reis, and Nigam~\cite{xavier.etal:2018} have formalized cut elimination and completeness of focusing for first-order linear logic in Coq, along with some other meta-theoretical properties.
In~\cite{dawson.gore:2010}, Dawson and Gor{\'e} describe their framework for formalizing sequent calculus with explicit structural rules in Isabelle/HOL.
They apply their framework for the provability logic GL and formalize the cut elimination argument for it.
Their framework was later ported Coq~\cite{dabrera.etal:2021} and used to formalize cut elimination for a modal logic Kt.
Another proof of cut elimination for GL was formalized in Coq~\cite{gore.etal:2021}; the authors noticed during the formalization process that the proof can be simplified in several parts.

\begin{sloppypar}
  Tews~\cite{tews:2013} used Coq to formalize Pattinson's and Schr{\"o}der's proof~\cite{pattinson.schroder:2010} of cut elimination for coalgebraic modal logics.
  During his formalization effort, Tews has uncovered a number of fixable gaps in the proof.
\end{sloppypar}

The formalized proofs mentioned above are syntactic.
A formalized semantic proof of cut elimination for the ${(\forall,\to,\bot)}$ fragment of intuitionistic FOL was given by Herbelin and Lee~\cite{herbelin.lee:2009}, using Kripke models.
The only similar formalization that we are aware is the formalization by Larchey-Wendling \cite{larchey_wendling:2021} of the Okada's semantic proof of cut elimination for linear logic~\cite{okada:99,okada:02}.
A similar formalization of cut elimination for implicational relevance logic was used by the author used part of a larger formalization~\cite{larchey-wendling:2020}.
In personal communication Larchey-Wendling mentioned that he has adapted the aforementioned phase semantics proof to the logic of Bunched Implications, but was not completely satisfied with it.

After Okada's proof, related methods for proving cut elimination were discussed for various logics.
For example, Belardinelli, Jipsen, and Ono~\cite{belardinelli:jipsen:ono:04} use intermediate structures (Gentzen structures) to interpret sequent calculi and prove cut elimination for various substructural variants of the Lambek calculus.
This method was generalized to handle non-associative logics (i.e. without the exchange rule)~\cite{galatos.ono:2010}.
Ciabattoni, Galatos, and Terui \cite{ciabattoni.etal:2008} prove semantic cut elimination for a wide ride of hypersequent calculi for nonclassical logics.

Galatos and Jipsen \cite{galatos.jipsen:2013} introduced the framework of residuated frames which they use to prove cut elimination (and other related properties) for many extensions of Lambek calculus with arbitrary structural rules.
The authors later extended their framework~\cite{galatos.jipsen:2017} to cover extensions of distributive Lambek calculus and BI.
\footnote{The residuated frames framework was used to derive other meta-theoretical properties, such as the finite model property.
Unfortunately, the finite model property proof in \cite{galatos.jipsen:2017} does not hold. The argument there relies on a version of the Curry's lemma (limiting a number of contractions that can occur in a given sequent in a proof search) which does not hold in BI (see \cite{jipsen2018algebraic}).}

Our proof can be seen as a simplification of the Galatos and Jipsen's method.
Instead of making heavy use of the residuated frames, our proof goes directly through algebraic semantics.
While this is a less general framework, it still allows us to extend the proof to cover, e.g. modal extensions of BI, which were not covered by the residuated frames framework.
We conjecture that the algebra we construct in \Cref{sec:cutelim_thm} is isomorphic to the Galois algebra constructed in \cite[Section 4]{galatos.jipsen:2017}.


\section{Conclusion and future work}
\label{sec:conclusion}
In this paper we have presented a fully formalized semantic-based proof of cut elimination for the logic of bunched implications.
We show that this proof can be extended to cover various extensions of BI, and demonstrated which parts of the proof have to be modified, and which remain unchanged.

As for future work, we see several ways of going forward.
Firstly, we can look at extensions of BI.
For example, we can probably extend the construction presented here to cover first-order/predicate BI.
The algebra $\mathcal{C}$ is already complete (has all the meets and joins), so it is suitable for interpreting quantifiers.
Unfortunately, formalizing this would require dealing with variable binders, which we decided to forgo in this paper.
It would also be natural to look at extensions such as GBI~\cite{{galatos.jipsen:2017}}, extensions of BI with various modalities that are used in separation logic~\cite{bizjak.birkedal:2018,dockins.etal:2008}, or the recently proposed polarized sequent calculus for BI~\cite{gheorghiu:marin:2021}.

\begin{sloppypar}
  Secondly, it would be interesting to go from logic to type theory.
  The algebra $\mathcal{C}$ is a subalgebra of predicates ${\Bunch \to \Prop}$, where $\Prop$ is the type of propositions.
  One can imagine it is possible to consider instead presheaves ${\Bunch \to \Set}$,
  and look for a categorification of ${\cal C}$ -- a reflexive subcategory of the category of presheaves, which is universal for cut-free provability.
  That might give us some insight into the connections to the normalization-by-evaluation method for type theories~\cite{altenkirch.etal:1995}, which is usually based on the category of presheaves.
\end{sloppypar}

\begin{acks}
The author would like to thank Jorge P{\'e}rez, Revantha Ramanayake, Niels van der Weide, and Dominique Larchey-Wendling, for their insightful comments on the earlier version of this article and for pointing me to some of the related work.
The author would also like to thank the anonymous CPP reviewers for providing their invaluable feedback.

\begin{sloppypar}
  The author was supported by VIDI Project No. 016.Vidi.189.046 (Unifying Correctness for Communicating Software).
\end{sloppypar}
\end{acks}

\bibliographystyle{ACM-Reference-Format}
\bibliography{cutelim}

\end{document}